\documentclass[draftclsnofoot, onecolumn, 10pt, one-side]{IEEEtran}
\usepackage{times,amsmath,color,amssymb,graphicx,epsfig,cite,psfrag,enumerate,amsthm,multicol,subfigure,multirow}

\theoremstyle{remark}
\newtheorem{Remark}{Remark}
\theoremstyle{plain}
\newtheorem{Theorem}{Theorem}

\newtheorem{Corollary}{Corollary}

\theoremstyle{definition}
\newtheorem{Definition}{Definition}

\allowdisplaybreaks

\begin{document}
\title{Interference Channel with State Information}
\author{Lili~Zhang, 
        Jinhua~Jiang, 
        and~Shuguang~Cui
\thanks{L. Zhang and S. Cui are with the Department
of Electrical and Computer Engineering, Texas A$\&$M University,
College Station, TX 77843 USA
(e-mail: lily.zhang@tamu.edu; cui@ece.tamu.edu).}
\thanks{J. Jiang is with the Department of Electrical Engineering, Stanford University, Stanford, CA 94305 USA (email: jhjiang@stanford.edu).}
}
\maketitle
\begin{abstract}
In this paper, we study the state-dependent two-user interference channel,
where the state information is non-causally known at both transmitters but
unknown to either of the receivers. We first propose two coding schemes for
the discrete memoryless case: simultaneous encoding for the sub-messages in
the first one and superposition encoding in the second one, both with rate
splitting and Gel'fand-Pinsker coding. The corresponding achievable rate
regions are established. Moreover, for the Gaussian case, we focus on the
simultaneous encoding scheme and propose an \emph{active interference
cancellation} mechanism, which is a generalized dirty-paper coding technique,
to partially eliminate the state effect at the receivers. The corresponding
achievable rate region is then derived. We also propose several heuristic
schemes for some special cases: the strong interference case, the mixed
interference case, and the weak interference case. For the strong and mixed
interference case, numerical results are provided to show that active
interference cancellation significantly enlarges the achievable rate region.
For the weak interference case, flexible power splitting instead of active
interference cancellation improves the performance significantly. 

\end{abstract}
\IEEEpeerreviewmaketitle
\section{Introduction} \label{sec_intro}
The interference channel (IC) models the situation where several independent
transmitters communicate with their corresponding receivers simultaneously
over a common spectrum. Due to the shared medium, each receiver suffers from
interferences caused by the transmissions of other transceiver pairs. The
research of IC was initiated by Shannon~\cite{shannon} and the channel was
first thoroughly studied by Ahlswede~\cite{ahlswede}. Later,
Carleial~\cite{carleial} established an improved achievable rate region by
applying the superposition coding scheme. In~\cite{han_kobayashi}, Han and
Kobayashi obtained the best achievable rate region known to date for the
general IC by utilizing simultaneous decoding at the receivers. Recently, this
rate region has been re-characterized with superposition encoding for the
sub-messages~\cite{cmg_region,new_fourier}. However, the capacity region of
the general IC is still an open problem~\cite{han_kobayashi}.


The capacity region for the corresponding Gaussian case is also unknown except
for several special cases, such as the strong Gaussian IC and the very strong
Gaussian IC~\cite{sato,carleial_2}. In addition, Sason~\cite{sason}
characterized the sum capacity for a special case of the Gaussian IC called
the degraded Gaussian IC. For more general cases, Han-Kobayashi
region~\cite{han_kobayashi} is still the best achievable rate region known to
date. However, for the general Gaussian interference channel, the calculation
of the Han-Kobayashi region bears high complexity. The authors in~\cite{david}
proposed a simpler heuristic coding scheme, for which they set the private
message power at both transmitters in a special way such that the interfered
private signal-to-noise ratio (SNR) at each receiver is equal to $1$. An upper
bound on the capacity was also derived in~\cite{david} and it was shown that
the gap between the heuristic lower bound and the capacity upper bound is less
than one bit for both weak and mixed interference cases.

Many variations of the interference channel have also been studied, including
the IC with feedback~\cite{jinhua} and the IC with conferencing
encoders/decoders~\cite{caoyi}. Here, we study another variation of the IC:
the state-dependent two-user IC with state information non-causally known at
both transmitters. This situation may arise in a multi-cell downlink
communication scenario as shown in Fig. \ref{fig_example}, where two
interested cells are interfering with each other and the mobiles suffer from
some common interference (which can be from other neighboring cells and viewed
as state) non-causally known at both of the two base-stations via certain
collaboration with the neighboring base-station. Notably, communication over
state-dependent channels has drawn lots of attentions due to its wide
applications such as information embedding~\cite{info_embedding} and computer
memories with defects~\cite{gamal}. The corresponding framework was also
initiated by Shannon in~\cite{shannon_2}, which established the capacity of a
state-dependent discrete memoryless (DM) point-to-point channel with causal
state information at the transmitter. In~\cite{gelfand}, Gel'fand and Pinsker
obtained the capacity for such a point-to-point case with the state
information non-causally known at the transmitter. Subsequently,
Costa~\cite{costa_dpc} extended Gel'fand-Pinsker coding to the state-dependent
additive white Gaussian noise (AWGN) channel, where the state is an additive
zero-mean Gaussian interference. This result is known as the dirty-paper
coding (DPC) technique, which achieves the capacity as if there is no such an
interference. For the multi-user case, extensions of the afore-mentioned
schemes appeared in~\cite{gelfand_2, kim, dpc_bc, mac_dpc_laneman} for the
multiple access channel (MAC), the broadcast channel, and the degraded
Gaussian relay channel, respectively.

In this paper, we study the state-dependent IC with state information
non-causally known at the transmitters and develop two coding schemes, both of
which jointly apply rate splitting and Gel'fand-Pinsker coding. In the first
coding scheme, we deploy simultaneous encoding for the sub-messages, and in
the second one, we deploy superposition encoding for the sub-messages. The
associated achievable rate regions are derived based on the respective coding
schemes. Then we specialize the achievable rate region corresponding to the
simultaneous encoding scheme in the Gaussian case, where the common additive
state is a zero-mean Gaussian random variable. Specifically, we introduce the
notion of \emph{active interference cancellation}, which generalizes
dirty-paper coding by utilizing some transmitting power to partially cancel
the common interference at both receivers. Furthermore, we propose heuristic
schemes for the strong Gaussian IC, the mixed Gaussian IC, and the weak
Gaussian IC with state information, respectively. For the strong Gaussian IC
with state information, the transmitters only send common messages and the DPC
parameters are optimized for one of the two resulting MACs. For the mixed
Gaussian IC with state information, one transmitter sends common message and
the other one sends private message, with DPC parameters optimized only for
one receiver. For the weak interference case, we apply rate
splitting, set the private message power at both transmitters to have the
interfered private SNR at each receiver equal to $1$~\cite{david}, utilize
sequential decoding, and optimize the DPC parameters for one of the MACs. The
time-sharing technique is applied in all the three cases to obtain enlarged
achievable rate regions. Numerical comparisons among the achievable rate
regions and the capacity outer bound are also provided. For the strong and
mixed interference cases, we show that the active interference cancellation
mechanism improves the performance significantly; for the weak interference
case, it is flexible power allocation instead of active interference
cancellation that enlarges the achievable rate region significantly. 

The rest of the paper is organized as follows. The channel model and the
definition of achievable rate region are presented in Section \ref{sec_2}. In
Section \ref{sec_3}, we provide two achievable rate regions for the discrete
memoryless IC with state information non-causally known at both transmitters,
based on the two different coding schemes, respectively. In Section
\ref{sec_4}, we discuss the Gaussian case and present the main idea of active
interference cancellation. The strong interference, mixed interference, and
weak interference cases are studied in Section \ref{sec_5}, \ref{sec_6}, and
\ref{sec_7}, respectively.
In Section~\ref{sec_9}, numerical results comparing different inner
bounds against the outer bound are given. Finally, we conclude the paper in
Section \ref{sec_conclusion}.
\begin{figure}[!t]
\centering
\rotatebox{90}{\includegraphics[width=2.5in]{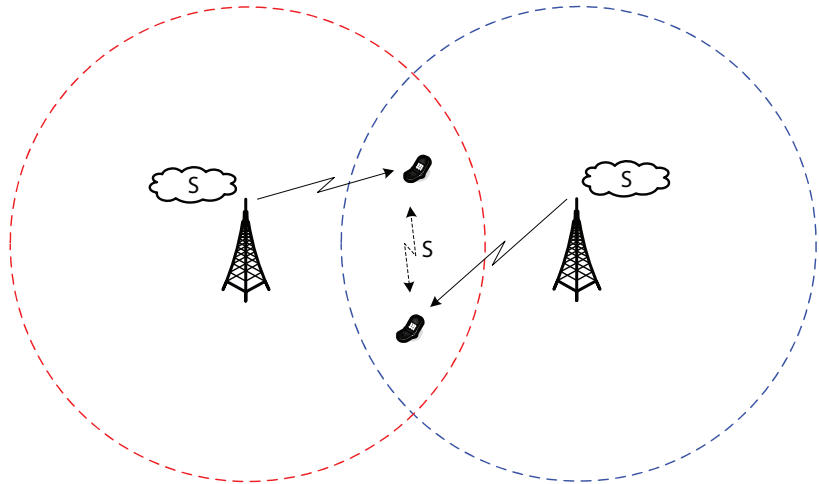}} \caption{A
multi-cell downlink communication example, which can be modeled as an
interference channel with state information non-causally known at both
transmitters.} \label{fig_example}
\end{figure}
\section{Channel Model}\label{sec_2}
Consider the interference channel as shown in
Fig.~\ref{fig_relay_channel_model}, where two transmitters communicate with
the corresponding receivers through a common medium that is dependent on state
$S$. The transmitters do not cooperate with each other; however, they both
know the state information $S$ non-causally, which is known to neither of the
receivers. Each receiver needs to decode the information from the
corresponding transmitter.
\subsection{Discrete Memoryless Case}\label{sec_2_DM}
We use the following notations for the DM channel. The random
variable is defined as $X$ with value $x$ in a finite set $\mathcal{X}$. Let
$p_{X}(x)$ be the probability mass function of $X$ on $\mathcal{X}$. The
corresponding sequences are denoted by $x^n$ with length $n$.
\begin{figure}[!t]
\centering
\includegraphics[width=3.5in]{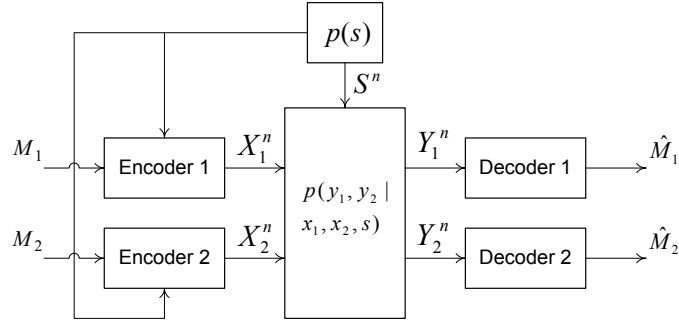}
\caption{The interference channel with state information non-causally known at
both transmitters.} \label{fig_relay_channel_model}
\end{figure}

The state-dependent two-user interference channel is defined by
$(\mathcal{X}_1,\mathcal{X}_2,\mathcal{Y}_1,\mathcal{Y}_2,\mathcal{S},p(y_1,y_2|x_1,x_2,s))$,
where $\mathcal{X}_1,\mathcal{X}_2$ are two input alphabet sets,
$\mathcal{Y}_1,\mathcal{Y}_2$ are the corresponding output alphabet sets,
$\mathcal{S}$ is the state alphabet set, and $p(y_1,y_2|x_1,x_2,s)$ is the
conditional probability of $(y_1,y_2) \in \mathcal{Y}_1 \mathcal{\times}
\mathcal{Y}_2$ given $(x_1,x_2,s)\in \mathcal{X}_1
\mathcal{\times}\mathcal{X}_2 \mathcal{\times}\mathcal{S}$. The channel is
assumed to be memoryless, i.e.,
\[
p(y_1^n,y_2^n|x_1^n,x_2^n,s^n) = \prod_{i=1}^{n} p(y_{1i},y_{2i}|x_{1i},x_{2i},s_{i}),
\]
where $i$ is the element index for each sequence.

A $(2^{nR_1},2^{nR_2},n)$ code for the above channel consists of two
independent message sets $\{1,2,\cdots,2^{nR_{1}}\}$ and
$\{1,2,\cdots,2^{nR_{2}}\}$, two encoders that respectively assign two
codewords to messages $m_1 \in \{1,2,\cdots,2^{nR_{1}}\}$ and $m_2 \in
\{1,2,\cdots,2^{nR_{2}}\}$ based on the non-causally known state information
$s^n$, and two decoders that respectively determine the estimated messages
$\hat{m}_1$ and $\hat{m}_2$ or declare an error from the received sequences.

The average probability of error is defined as:
\begin{equation}
P_e^{(n)} = \frac{1}{2^{n(R_1+R_2)}}\sum_{m_1,m_2} \textrm{Pr}\{\hat{m}_1 \neq m_1 \textrm{ or } \hat{m}_2 \neq m_2|(m_1,m_2) \textrm{ is sent}\},
\end{equation}
where $(m_1,m_2)$ is assumed to be uniformly distributed over
$\{1,2,\cdots,2^{nR_{1}}\}\times\{1,2,\cdots,2^{nR_{2}}\}$.

\begin{Definition}
A rate pair $(R_1,R_2)$ of non-negative real values is achievable if there
exists a sequence of $(2^{nR_1},2^{nR_2},n)$ codes with $P_e^{(n)}\to 0$ as $n
\to \infty$. The set of all achievable rate pairs is defined as the capacity
region.
\end{Definition}
\subsection{Gaussian Case}\label{sec_2_Gaussian}
\begin{figure}[!t]
\centering
\rotatebox{90}{\includegraphics[width=2.5in]{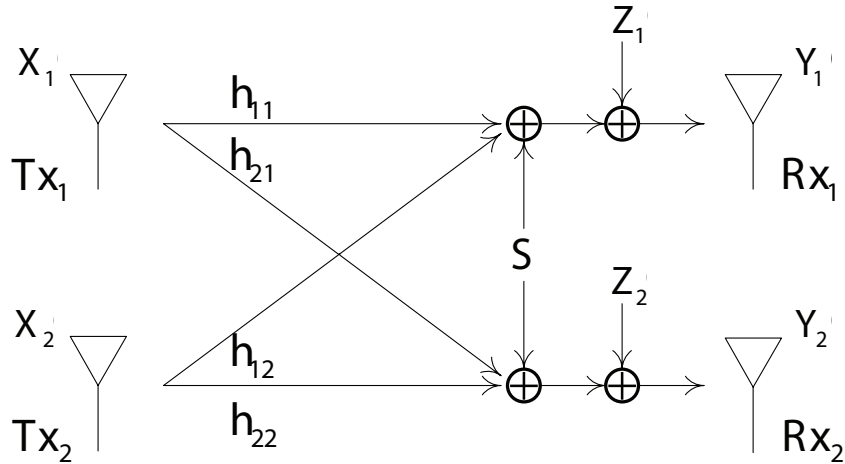}}
\caption{The Gaussian interference channel with state information non-causally
known at both transmitters.} \label{fig_channel_model}
\end{figure}
The Gaussian counterpart of the previously defined DM channel is shown in Fig.
\ref{fig_channel_model}, where two transmitters communicate with the
corresponding receivers through a common channel that is dependent on state
$S$, which can be treated as a common interference. The corresponding signal
structure can be described by the following channel input and output
relationship:
\begin{eqnarray*}
Y_1' &=& h_{11}X_1'+h_{12}X_2'+S+Z_1',\\
Y_2' &=& h_{22}X_2'+h_{21}X_1'+S+Z_2',
\end{eqnarray*}
where $h_{ij}$ is the real link amplitude gain from the $j$th transmitter to
the $i$th receiver, $X_i'$ and $Y_i'$ are the channel input and output,
respectively, and $Z_i'$ is the zero-mean AWGN noise with variance $N_i$, for
$i=1,2$ and $j=1,2$. Both receivers also suffer from a zero-mean additive
white Gaussian interference $S$ with variance $K$, which is non-causally known
at both transmitters\footnotemark. Note that for this AWGN
model, all the random variables are defined over the field of real numbers
$\mathbb{R}$.
\footnotetext[1]{In general, the additive
states over the two links may not be the same, i.e., we may have $Y_1' =
h_{11}X_1'+h_{12}X_2'+S_1+Z_1'$ and $Y_2' = h_{22}X_2'+h_{21}X_1'+S_2+Z_2'$.
However, in this paper we only focus on the simplest scenario: $S_1=S_2=S$.
The more general cases with state $(S_1,S_2)$ and different knowledge levels at the two transmitters will be studied in our future work.}

Without loss of generality, we transform the signal model into the following
standard form~\cite{han_kobayashi}:
\begin{eqnarray}\label{eq_gaussian_model1}
Y_1 &=& X_1 + \sqrt{g_{12}}X_2 + \frac{1}{\sqrt{N_1}}S+Z_1,\\
Y_2 &=& X_2 + \sqrt{g_{21}}X_1 + \frac{1}{\sqrt{N_2}}S+Z_2,\label{eq_gaussian_model2}
\end{eqnarray}
where
\begin{eqnarray*}
Y_1 &=& \frac{Y_1'}{\sqrt{N_1}},\ X_1 = \frac{h_{11}X_1'}{\sqrt{N_1}},\ g_{12} = \frac{h_{12}^2N_2}{h_{22}^2N_1},\ Z_1 = \frac{Z_1'}{\sqrt{N_1}},\\
Y_2 &=& \frac{Y_2'}{\sqrt{N_2}},\ X_2 = \frac{h_{22}X_2'}{\sqrt{N_2}},\ g_{21} = \frac{h_{21}^2N_1}{h_{11}^2N_2},\ Z_2 = \frac{Z_2'}{\sqrt{N_2}}.
\end{eqnarray*}
Note that $Z_1$ and $Z_2$ have unit variance in \eqref{eq_gaussian_model1} and
\eqref{eq_gaussian_model2}. We also impose the following power constraints on
the channel inputs $X_1$ and $X_2$:
\[
\frac{1}{n}\sum_{i=1}^{n}(X_{1i})^2\leq P_1,\textrm{ and }\frac{1}{n}\sum_{i=1}^{n}(X_{2i})^2\leq P_2.
\]
\section{Achievable Rate Regions for DM Interference Channel with State Information} \label{sec_3}
In this section, we propose two new coding schemes for the DM interference
channel with state information non-causally known at both transmitters and
quantify the associated achievable rate regions. For both coding schemes, we
jointly deploy rate splitting and Gel'fand-Pinsker coding. Specifically, in
the first coding scheme, we use simultaneous encoding on the sub-messages,
while in the second one we apply superposition encoding.
\subsection{Coding Scheme I: Simultaneous Encoding}\label{sec_3_simul}
Now we introduce the following rate region achieved by the first coding
scheme, which combines rate splitting and Gel'fand-Pinsker coding. Let us
consider the auxiliary random variables $Q$, $U_1$, $V_1$, $U_2$, and $V_2$,
defined on arbitrary finite sets $\mathcal{Q}$, $\mathcal{U}_1$,
$\mathcal{V}_1$, $\mathcal{U}_2$, and $\mathcal{V}_2$, respectively. The joint
probability distribution of the above auxiliary random variables and the state
variable $S$ is chosen to satisfy the form
$p(s)p(q)p(u_1|q,s)p(v_1|q,s)p(u_2|q,s)p(v_2|q,s)$. Moreover, for a given $Q$,
we let the channel input $X_j$ be an arbitrary deterministic function of
$U_j$, $V_j$, and $S$. The achievable rate region of the simultaneous encoding
scheme is given in the following theorem.
\begin{Theorem}\label{theorem_1}
For a fixed probability distribution
$p(q)p(u_1|q,s)p(v_1|q,s)p(u_2|q,s)p(v_2|q,s)$, let $\mathcal{R}_1$ be the set
of all non-negative rate tuple $(R_{10},R_{11},R_{20},R_{22})$ satisfying
\setlength{\arraycolsep}{2pt}{\small
\begin{eqnarray}
\label{eq_rate_constraint_11} R_{11} &\leq& I(U_1;U_2|Q)+I(U_1,U_2;V_1|Q)+I(V_1;Y_1|U_1,U_2,Q)-I(V_1;S|Q),\\
\label{eq_rate_constraint_12} R_{10} &\leq& I(U_1;U_2|Q)+I(U_1,U_2;V_1|Q)+I(U_1;Y_1|V_1,U_2,Q)-I(U_1;S|Q),\\
\label{eq_rate_constraint_13} R_{10}+R_{11} &\leq& I(U_1;U_2|Q)+I(U_1,U_2;V_1|Q)+I(U_1,V_1;Y_1|U_2,Q)-I(U_1;S|Q)-I(V_1;S|Q),\\
\label{eq_rate_constraint_14} R_{11}+R_{20} &\leq& I(U_1;U_2|Q)+I(U_1,U_2;V_1|Q)+I(V_1,U_2;Y_1|U_1,Q)-I(V_1;S|Q)-I(U_2;S|Q),\\
\label{eq_rate_constraint_15} R_{10}+R_{20} &\leq& I(U_1;U_2|Q)+I(U_1,U_2;V_1|Q)+I(U_1,U_2;Y_1|V_1,Q)-I(U_1;S|Q)-I(U_2;S|Q),\\
\label{eq_rate_constraint_16} R_{10}+R_{11}+R_{20} &\leq& I(U_1;U_2|Q)+I(U_1,U_2;V_1|Q)+I(U_1,V_1,U_2;Y_1|Q)-I(U_1;S|Q)-I(V_1;S|Q)-I(U_2;S|Q),\\
\label{eq_rate_constraint_21} R_{22} &\leq& I(U_2;U_1|Q)+I(U_2,U_1;V_2|Q)+I(V_2;Y_2|U_2,U_1,Q)-I(V_2;S|Q),\\
\label{eq_rate_constraint_22} R_{20} &\leq& I(U_2;U_1|Q)+I(U_2,U_1;V_2|Q)+I(U_2;Y_2|V_2,U_1,Q)-I(U_2;S|Q),\\
\label{eq_rate_constraint_23} R_{20}+R_{22} &\leq& I(U_2;U_1|Q)+I(U_2,U_1;V_2|Q)+I(U_2,V_2;Y_2|U_1,Q)-I(U_2;S|Q)-I(V_2;S|Q),\\
\label{eq_rate_constraint_24} R_{22}+R_{10} &\leq& I(U_2;U_1|Q)+I(U_2,U_1;V_2|Q)+I(V_2,U_1;Y_2|U_2,Q)-I(V_2;S|Q)-I(U_1;S|Q),\\
\label{eq_rate_constraint_25} R_{20}+R_{10} &\leq& I(U_2;U_1|Q)+I(U_2,U_1;V_2|Q)+I(U_2,U_1;Y_2|V_2,Q)-I(U_2;S|Q)-I(U_1;S|Q),\\
\label{eq_rate_constraint_26} R_{20}+R_{22}+R_{10} &\leq& I(U_2;U_1|Q)+I(U_2,U_1;V_2|Q)+I(U_2,V_2,U_1;Y_2|Q)-I(U_2;S|Q)-I(V_2;S|Q)-I(U_1;S|Q).
\end{eqnarray}}
Then for any $(R_{10},R_{11},R_{20},R_{22}) \in \mathcal{R}_1$, the rate pair
$(R_{10}+R_{11},R_{20}+R_{22})$ is achievable for the DM interference channel
with state information defined in Section~\ref{sec_2}.
\end{Theorem}

\begin{Remark}
The detailed proof is given in Appendix~\ref{appendix_1} with the outline
sketched as follows. For the coding scheme in Theorem \ref{theorem_1}, the
message at transmitter $j$ ($j=1$ or $2$) is splitted into two parts: the
public message $m_{j0}$ and the private message $m_{jj}$. Furthermore,
Gel'fand-Pinsker coding is utilized to help both transmitters send the
messages with the non-causal knowledge of the state information. Specifically,
transmitter $j$ finds the corresponding public codeword $u_j$ and the private
codeword $v_j$ such that they are jointly typical with the state $s^n$. Then
the transmitting codeword is constructed as a deterministic
function of the public codeword $u_j$, the private codeword $v_j$, and the
state $s^n$. At the receiver side, decoder $j$ tries to decode the
corresponding messages from transmitter $j$ and the public message of the
interfering transmitter. The rest follows by the usual error event grouping
and error probability analysis.
\end{Remark}
\begin{Remark}
The auxiliary random variables in Theorem~\ref{theorem_1} can be interpreted
as follows: $Q$ is the time-sharing random variable; $U_j$ and $V_j$ ($j=1$ or
$2$) are the auxiliary random variables to carry the public and private
messages at transmitter $j$, respectively. It can be easily seen from the
joint probability distribution that $U_j$ and $V_j$ are conditionally
independent given $Q$ and $S$, which means that the public and private
messages are encoded ``simultaneously".
\end{Remark}

An explicit description of the achievable rate region can be obtained by
applying the Fourier-Motzkin algorithm~\cite{cmg_region} on our implicit
description \eqref{eq_rate_constraint_11}-\eqref{eq_rate_constraint_26}, as
shown in the next corollary.
\begin{Corollary}
For a fixed probability distribution
$p(q)p(u_1|q,s)p(v_1|q,s)p(u_2|q,s)p(v_2|q,s)$, let $\hat{\mathcal{R}}_{1}$ be
the set of all non-negative rate pairs $(R_{1},R_{2})$ satisfying
\setlength{\arraycolsep}{1pt}{\small
\begin{eqnarray}
\label{eq_ex_rate_constraint_1} R_{1} &\leq& \min\{d_1,g_1,a_1+b_1,a_1+f_1,a_1+e_2,a_1+f_2,b_1+e_1,e_1+f_1,e_1+f_2\},\\
\label{eq_ex_rate_constraint_2} R_{2} &\leq& \min\{d_2,g_2,a_2+b_2,a_2+f_2,a_2+e_1,a_2+f_1,b_2+e_2,e_2+f_2,e_2+f_1\},\\
\label{eq_ex_rate_constraint_3} R_1+R_{2} &\leq& \min\{a_1+g_2,a_2+g_1,e_1+g_2,e_2+g_1,e_1+e_2,a_1+a_2+f_1,a_1+a_2+f_2,a_1+b_2+e_2,a_2+b_1+e_1\},\\
\label{eq_ex_rate_constraint_4} R_{1}+2R_{2} &\leq& \min\{e_1+f_1+2a_2,e_1+2a_2+f_2,e_1+a_2+g_2\},\\
\label{eq_ex_rate_constraint_5} 2R_1+R_{2} &\leq& \min\{e_2+f_2+2a_1,e_2+2a_1+f_1,e_2+a_1+g_1\},
\end{eqnarray}}
where
\begin{eqnarray*}
a_1 &=& I(U_1;U_2|Q)+I(U_1,U_2;V_1|Q)+I(V_1;Y_1|U_1,U_2,Q)-I(V_1;S|Q),\\
b_1 &=& I(U_1;U_2|Q)+I(U_1,U_2;V_1|Q)+I(U_1;Y_1|V_1,U_2,Q)-I(U_1;S|Q),\\
d_1 &=& I(U_1;U_2|Q)+I(U_1,U_2;V_1|Q)+I(U_1,V_1;Y_1|U_2,Q)-I(U_1;S|Q)-I(V_1;S|Q),\\
e_1 &=& I(U_1;U_2|Q)+I(U_1,U_2;V_1|Q)+I(V_1,U_2;Y_1|U_1,Q)-I(V_1;S|Q)-I(U_2;S|Q),\\
f_1 &=& I(U_1;U_2|Q)+I(U_1,U_2;V_1|Q)+I(U_1,U_2;Y_1|V_1,Q)-I(U_1;S|Q)-I(U_2;S|Q),\\
g_1 &=& I(U_1;U_2|Q)+I(U_1,U_2;V_1|Q)+I(U_1,V_1,U_2;Y_1|Q)-I(U_1;S|Q)-I(V_1;S|Q)-I(U_2;S|Q),\\
a_2 &=& I(U_2;U_1|Q)+I(U_2,U_1;V_2|Q)+I(V_2;Y_2|U_2,U_1,Q)-I(V_2;S|Q),\\
b_2 &=& I(U_2;U_1|Q)+I(U_2,U_1;V_2|Q)+I(U_2;Y_2|V_2,U_1,Q)-I(U_2;S|Q),\\
d_2 &=& I(U_2;U_1|Q)+I(U_2,U_1;V_2|Q)+I(U_2,V_2;Y_2|U_1,Q)-I(U_2;S|Q)-I(V_2;S|Q),\\
e_2 &=& I(U_2;U_1|Q)+I(U_2,U_1;V_2|Q)+I(V_2,U_1;Y_2|U_2,Q)-I(V_2;S|Q)-I(U_1;S|Q),\\
f_2 &=& I(U_2;U_1|Q)+I(U_2,U_1;V_2|Q)+I(U_2,U_1;Y_2|V_2,Q)-I(U_2;S|Q)-I(U_1;S|Q),\\
g_2 &=& I(U_2;U_1|Q)+I(U_2,U_1;V_2|Q)+I(U_2,V_2,U_1;Y_2|Q)-I(U_2;S|Q)-I(V_2;S|Q)-I(U_1;S|Q).
\end{eqnarray*}
Then any rate pair $(R_{1},R_{2}) \in \hat{\mathcal{R}}_1$ is achievable for
the DM interference channel with state information defined in
Section~\ref{sec_2}.
\end{Corollary}
\subsection{Coding Scheme II: Superposition Encoding}\label{sec_3_superposition}
We now present the second coding scheme, which applies superposition encoding
for the sub-messages. Similar to the auxiliary random variables in
Theorem~\ref{theorem_1}, in the following theorem, $Q$ is also the
time-sharing random variable; $U_j$ and $V_j$ ($j=1$ or $2$) are the auxiliary
random variables to carry the public and private messages at transmitter $j$,
respectively. The difference here is the joint probability distribution
$p(s)p(q)p(u_1|s,q)p(v_1|u_1,s,q)p(u_2|s,q)p(v_2|u_2,s,q)$, where $U_j$ and
$V_j$ are not conditionally independent given $Q$ and $S$. This also implies
the notion of ``superposition encoding". The achievable rate region of the
superposition encoding scheme is given in the following theorem.
\begin{Theorem}\label{theorem_2}
For a fixed probability distribution $p(q)p(u_1|s,q)p(v_1|u_1,s,q)p(u_2|s,q)p(v_2|u_2,s,q)$, let $\mathcal{R}_2$ be the set of all non-negative rate tuple $(R_{10},R_{11},R_{20},R_{22})$ satisfying
\begin{eqnarray}
\label{eq_rate_constraint_2_11} R_{11} &\leq& I(U_1,V_1;U_2|Q)+I(V_1;Y_1|U_1,U_2,Q)-I(V_1;S|U_1,Q),\\
\label{eq_rate_constraint_2_12} R_{10}+R_{11} &\leq& I(U_1,V_1;U_2|Q)+I(U_1,V_1;Y_1|U_2,Q)-I(U_1,V_1;S|Q),\\
\label{eq_rate_constraint_2_13} R_{11}+R_{20} &\leq& I(U_1,V_1;U_2|Q)+I(V_1,U_2;Y_1|U_1,Q)-I(V_1;S|U_1,Q)-I(U_2;S|Q),\\
\label{eq_rate_constraint_2_14} R_{10}+R_{11}+R_{20} &\leq& I(U_1,V_1;U_2|Q)+I(U_1,V_1,U_2;Y_1|Q)-I(U_1,V_1;S|Q)-I(U_2;S|Q),\\
\label{eq_rate_constraint_2_21} R_{22} &\leq& I(U_2,V_2;U_1|Q)+I(V_2;Y_2|U_2,U_1,Q)-I(V_2;S|U_2,Q),\\
\label{eq_rate_constraint_2_22} R_{20}+R_{22} &\leq& I(U_2,V_2;U_1|Q)+I(U_2,V_2;Y_2|U_1,Q)-I(U_2,V_2;S|Q),\\
\label{eq_rate_constraint_2_23} R_{22}+R_{10} &\leq& I(U_2,V_2;U_1|Q)+I(V_2,U_1;Y_2|U_2,Q)-I(V_2;S|U_2,Q)-I(U_1;S|Q),\\
\label{eq_rate_constraint_2_24} R_{20}+R_{22}+R_{10} &\leq& I(U_2,V_2;U_1|Q)+I(U_2,V_2,U_1;Y_2|Q)-I(U_2,V_2;S|Q)-I(U_1;S|Q).
\end{eqnarray}
Then for any $(R_{10},R_{11},R_{20},R_{22}) \in  \mathcal{R}_2$, the rate pair
$(R_{10}+R_{11},R_{20}+R_{22})$ is achievable for the DM interference channel
with state information defined in Section \ref{sec_2}.
\end{Theorem}
The detailed proof for Theorem \ref{theorem_2} is given in Appendix \ref{appendix_2}.
\begin{Remark}
Compared with the first coding scheme in Theorem \ref{theorem_1}, the rate
splitting structure is also applied in the achievable scheme of Theorem
\ref{theorem_2}. The main difference here is that instead of simultaneous
encoding, now the private message $m_{jj}$ is superimposed on the public
message $m_{j0}$ for the $j$th transmitter, $j=1$, $2$. In addition,
Gel'fand-Pinsker coding is utilized to help the transmitters send both public
and private messages.
\end{Remark}
\begin{Remark}
It can be easily seen that the achievable rate region $\mathcal{R}_1$ in
Theorem \ref{theorem_1} is a subset of $\mathcal{R}_2$, i.e., $\mathcal{R}_1
\subseteq \mathcal{R}_2$. However, whether these two regions can be equivalent
is still under investigation, which is motivated by the equivalence between
the simultaneous encoding region and the superposition encoding region for the
traditional IC~\cite{cmg_region}.
\end{Remark}
\section{The Gaussian Interference Channel with State Information}\label{sec_4}
In this section, we present the corresponding achievable rate region for the
Gaussian IC with state information defined in Section \ref{sec_2}. In addition
to applying dirty paper coding and rate splitting, here we also introduce the
idea of active interference cancellation, which allocates some source power to
cancel the state effect at the receivers.
\subsection{Active Interference Cancellation}\label{sec_4_active_cancel}
In the general Gaussian interference channel, the simultaneous encoding over
the sub-messages can be viewed as sending $X_j = A_j+B_j$ at the $j$th
transmitter, $j=1$, $2$, where $A_j$ and $B_j$ are independent and correspond
to the public and private messages, respectively. Correspondingly, for the
Gaussian IC with state information defined in Section \ref{sec_2}, we focus on
the coding scheme based on simultaneous encoding that was discussed in
Section~\ref{sec_3_simul}. Specifically, we apply dirty paper coding to both
public and private parts, i.e., we define the auxiliary variables as follows:
\begin{eqnarray}
U_1&=&A_1+\alpha_{10}S,\ V_1 = B_1+\alpha_{11}S,\\
U_2&=&A_2+\alpha_{20}S,\ V_2 = B_2+\alpha_{22}S.
\end{eqnarray}

In addition, we allow both transmitters to apply active interference
cancellation by allocating a certain amount of power to send counter-phase
signals against the known interference $S$, i.e.,
\begin{eqnarray}
X_1&=&A_1+B_1-\gamma_1S,
\\
X_2&=&A_2+B_2-\gamma_2S,
\end{eqnarray}
where $\gamma_1$ and $\gamma_2$ are active cancellation parameters. The idea is
to generalize dirty-paper coding by allocating some transmitting power to
cancel part of the state effect at both receivers. Assume
$A_1\sim\mathcal{N}(0,\beta_1(P_1-\gamma_1^2 K))$,
$B_1\sim\mathcal{N}(0,\bar{\beta}_1(P_1-\gamma_1^2 K))$,
$A_2\sim\mathcal{N}(0,\beta_2(P_2-\gamma_2^2 K))$, and
$B_2\sim\mathcal{N}(0,\bar{\beta}_2(P_2-\gamma_2^2 K))$, where
$\beta_1+\bar{\beta}_1=1$ and $\beta_2+\bar{\beta}_2=1$. According to the
Gaussian channel model defined in Section \ref{sec_2}, the received signals
can be determined as:
\begin{eqnarray*}
Y_1 &=& A_1+B_1+\sqrt{g_{12}}(A_2+B_2)+\mu_{1}S+Z_1,\\
Y_2 &=& A_2+B_2+\sqrt{g_{21}}(A_1+B_1)+\mu_{2}S+Z_2,
\end{eqnarray*}
where $\mu_{1}=\frac{1}{\sqrt{N_1}}-\gamma_1-\gamma_2\sqrt{g_{12}}$ and
$\mu_{2}=\frac{1}{\sqrt{N_2}}-\gamma_2-\gamma_1\sqrt{g_{21}}$.

For convenience, we denote $P_{A_1}=\beta_1(P_1-\gamma_1^2 K)$,
$P_{B_1}=\bar{\beta}_1(P_1-\gamma_1^2 K)$, $P_{A_2}=\beta_2(P_2-\gamma_2^2
K)$, and $P_{B_2}=\bar{\beta}_2(P_2-\gamma_2^2 K)$. Also define
$G_{U_1}=\alpha_{10}^2K/P_{A_1}$, $G_{V_1}=\alpha_{11}^2K/P_{B_1}$,
$G_{U_2}=\alpha_{20}^2K/P_{A_2}$, and $G_{V_2}=\alpha_{22}^2K/P_{B_2}$.
\subsection{Achievable Rate Region}\label{sec_4_rate_region}
The achievable rate region can be obtained by evaluating the rate region given
in Theorem \ref{theorem_1} with respect to the corresponding Gaussian
auxiliary variables and channel outputs.
\begin{Theorem}
Let $\mathcal{R}_{1}'$ be the set of all non-negative rate tuple
$(R_{10},R_{11},R_{20},R_{22})$ satisfying
\setlength{\arraycolsep}{1pt}{\small\begin{eqnarray*} R_{11} &\leq&
\frac{1}{2}\log\left(\frac{\left(1+P_{B_1}+g_{12}P_{B_2}\right)\left(1+G_{U_1}+G_{U_2}+G_{U_1}G_{U_2}\right)
+K(\alpha_{10}+\alpha_{20}\sqrt{g_{12}}-\mu_{1})^2\left(1+\frac{G_{U_1}G_{U_2}}
{1+G_{U_1}+G_{U_2}}\right)}{\left(1+g_{12}P_{B_2}\right)\left(1+G_{U_1}+G_{U_2}+G_{V_1}\right)+K(\alpha_{10}+\alpha_{20}\sqrt{g_{12}}+\alpha_{11}-\mu_{1})^2}\right),\\
R_{10} &\leq& \frac{1}{2}\log\left(\frac{\left(1+P_{A_1}+g_{12}P_{B_2}\right)\left(1+G_{V_1}+G_{U_2}+G_{V_1}G_{U_2}\right)+K(\alpha_{11}+\alpha_{20}\sqrt{g_{12}}-\mu_{1})^2\left(1+\frac{G_{V_1}G_{U_2}}
{1+G_{V_1}+G_{U_2}}\right)}{\left(1+g_{12}P_{B_2}\right)\left(1+G_{U_1}+G_{U_2}+G_{V_1}\right)+K(\alpha_{10}+\alpha_{20}\sqrt{g_{12}}+\alpha_{11}-\mu_{1})^2}\right),\\
R_{10}+R_{11} &\leq& \frac{1}{2}\log\left(\frac{\left(1+P_{A_1}+P_{B_1}+g_{12} P_{B_2}\right)\left(1+G_{U_2}\right)+K(\alpha_{20}\sqrt{g_{12}}-\mu_{1})^2}{\left(1+g_{12}P_{B_2}\right)\left(1+G_{U_1}+G_{U_2}+G_{V_1}\right)+K(\alpha_{10}+\alpha_{20}\sqrt{g_{12}}+\alpha_{11}-\mu_{1})^2}\right),\\
R_{11}+R_{20} &\leq& \frac{1}{2}\log\left(\frac{\left(1+P_{B_1}+g_{12}P_{A_2}+g_{12} P_{B_2}\right)\left(1+G_{U_1}\right)+K(\alpha_{10}-\mu_{1})^2}{\left(1+g_{12}P_{B_2}\right)\left(1+G_{U_1}+G_{U_2}+G_{V_1}\right)+K(\alpha_{10}+\alpha_{20}\sqrt{g_{12}}+\alpha_{11}-\mu_{1})^2}\right),\\
R_{10}+R_{20} &\leq& \frac{1}{2}\log\left(\frac{\left(1+P_{A_1}+g_{12}P_{A_2}+g_{12} P_{B_2}\right)\left(1+G_{V_1}\right)+K(\alpha_{11}-\mu_{1})^2}{\left(1+g_{12}P_{B_2}\right)\left(1+G_{U_1}+G_{U_2}+G_{V_1}\right)+K(\alpha_{10}+\alpha_{20}\sqrt{g_{12}}+\alpha_{11}-\mu_{1})^2}\right),\\
R_{10}+R_{11}+R_{20} &\leq& \frac{1}{2}\log\left(\frac{1+P_{A_1}+P_{B_1}+g_{12}P_{A_2}+g_{12}P_{B_2}+\mu_{1}^2 K}{\left(1+g_{12}P_{B_2}\right)\left(1+G_{U_1}+G_{U_2}+G_{V_1}\right)+K(\alpha_{10}+\alpha_{20}\sqrt{g_{12}}+\alpha_{11}-\mu_{1})^2}\right),\\
R_{22} &\leq& \frac{1}{2}\log\left(\frac{\left(1+P_{B_2}+g_{21}P_{B_1}\right)\left(1+G_{U_2}+G_{U_1}+G_{U_2}G_{U_1}\right)
+K(\alpha_{20}+\alpha_{10}\sqrt{g_{21}}-\mu_{2})^2\left(1+\frac{G_{U_2}G_{U_1}}
{1+G_{U_2}+G_{U_1}}\right)}{\left(1+g_{21}P_{B_1}\right)\left(1+G_{U_2}+G_{U_1}+G_{V_2}\right)+K(\alpha_{20}+\alpha_{10}\sqrt{g_{21}}+\alpha_{22}-\mu_{2})^2}\right),\\
R_{20} &\leq& \frac{1}{2}\log\left(\frac{\left(1+P_{A_2}+g_{21}P_{B_1}\right)\left(1+G_{V_2}+G_{U_1}+G_{V_2}G_{U_1}\right)+K(\alpha_{22}+\alpha_{10}\sqrt{g_{21}}-\mu_{2})^2\left(1+\frac{G_{V_2}G_{U_1}}
{1+G_{V_2}+G_{U_1}}\right)}{\left(1+g_{21}P_{B_1}\right)\left(1+G_{U_2}+G_{U_1}+G_{V_2}\right)+K(\alpha_{20}+\alpha_{10}\sqrt{g_{21}}+\alpha_{22}-\mu_{2})^2}\right),\\
R_{20}+R_{22} &\leq& \frac{1}{2}\log\left(\frac{\left(1+P_{A_2}+P_{B_2}+g_{21} P_{B_1}\right)\left(1+G_{U_1}\right)+K(\alpha_{10}\sqrt{g_{21}}-\mu_{2})^2}{\left(1+g_{21}P_{B_1}\right)\left(1+G_{U_2}+G_{U_1}+G_{V_2}\right)+K(\alpha_{20}+\alpha_{10}\sqrt{g_{21}}+\alpha_{22}-\mu_{2})^2}\right),\\
R_{22}+R_{10} &\leq& \frac{1}{2}\log\left(\frac{\left(1+P_{B_2}+g_{21}P_{A_1}+g_{21} P_{B_1}\right)\left(1+G_{U_2}\right)+K(\alpha_{20}-\mu_{2})^2}{\left(1+g_{21}P_{B_1}\right)\left(1+G_{U_2}+G_{U_1}+G_{V_2}\right)+K(\alpha_{20}+\alpha_{10}\sqrt{g_{21}}+\alpha_{22}-\mu_{2})^2}\right),\\
R_{20}+R_{10} &\leq& \frac{1}{2}\log\left(\frac{\left(1+P_{A_2}+g_{21}P_{A_1}+g_{21} P_{B_1}\right)\left(1+G_{V_2}\right)+K(\alpha_{22}-\mu_{2})^2}{\left(1+g_{21}P_{B_1}\right)\left(1+G_{U_2}+G_{U_1}+G_{V_2}\right)+K(\alpha_{20}+\alpha_{10}\sqrt{g_{21}}+\alpha_{22}-\mu_{2})^2}\right),\\
R_{20}+R_{22}+R_{10} &\leq&
\frac{1}{2}\log\left(\frac{1+P_{A_2}+P_{B_2}+g_{21}P_{A_1}+g_{21}P_{B_1}+\mu_{2}^2
K}{\left(1+g_{21}P_{B_1}\right)\left(1+G_{U_2}+G_{U_1}+G_{V_2}\right)+K(\alpha_{20}+\alpha_{10}\sqrt{g_{21}}+\alpha_{22}-\mu_{2})^2}\right).
\end{eqnarray*}}
Then for any $(R_{10},R_{11},R_{20},R_{22}) \in \mathcal{R}_1'$, the rate pair
$(R_{10}+R_{11},R_{20}+R_{22})$ is achievable for the Gaussian IC with state
information defined in Section \ref{sec_2}.\label{theorem_gaussian_1}
\end{Theorem}
Note that the achievable rate region $\mathcal{R}_1'$ depends on the power
splitting parameters, the active cancellation parameters, and the DPC
parameters. To be clear, we may write $\mathcal{R}_1'$ as
$\mathcal{R}_1'\left(\beta_1,\beta_2,\gamma_1,\gamma_2,\alpha_{10},\alpha_{11},\alpha_{20},\alpha_{22}\right)$.
\begin{Remark}
It can be easily seen that the above achievable rate region includes the
capacity region of the Gaussian MAC with state information, by only using the
common messages for both transmitters and optimizing the respective DPC
parameters.
\end{Remark}

The following corollary gives the achievable rate region for the Gaussian IC
with state information when the state power $K\to\infty$.
\begin{Corollary}
Let $\widetilde{\mathcal{R}}_{1}'$ be the set of all non-negative rate tuple $(R_{10},R_{11},R_{20},R_{22})$ satisfying
\begin{eqnarray*}
R_{10}+R_{11} &\leq& \frac{1}{2}\log\left(\frac{\left(1+P_{A_1}+P_{B_1}+g_{12} P_{B_2}\right)\frac{\alpha_{20}^2}{P_{A_2}}+(\alpha_{20}\sqrt{g_{12}}-\mu_{1})^2}{\left(1+g_{12}P_{B_2}\right)\left(\frac{\alpha_{10}^2}{P_{A_1}}+\frac{\alpha_{20}^2}{P_{A_2}}+\frac{\alpha_{11}^2}{P_{B_1}}\right)+(\alpha_{10}+\alpha_{20}\sqrt{g_{12}}+\alpha_{11}-\mu_{1})^2}\right),\\
R_{11}+R_{20} &\leq& \frac{1}{2}\log\left(\frac{\left(1+P_{B_1}+g_{12}P_{A_2}+g_{12} P_{B_2}\right)\frac{\alpha_{10}^2}{P_{A_1}}+(\alpha_{10}-\mu_{1})^2}{\left(1+g_{12}P_{B_2}\right)\left(\frac{\alpha_{10}^2}{P_{A_1}}+\frac{\alpha_{20}^2}{P_{A_2}}+\frac{\alpha_{11}^2}{P_{B_1}}\right)+(\alpha_{10}+\alpha_{20}\sqrt{g_{12}}+\alpha_{11}-\mu_{1})^2}\right),\\
R_{10}+R_{20} &\leq& \frac{1}{2}\log\left(\frac{\left(1+P_{A_1}+g_{12}P_{A_2}+g_{12} P_{B_2}\right)\frac{\alpha_{11}^2}{P_{B_1}}+(\alpha_{11}-\mu_{1})^2}{\left(1+g_{12}P_{B_2}\right)\left(\frac{\alpha_{10}^2}{P_{A_1}}+\frac{\alpha_{20}^2}{P_{A_2}}+\frac{\alpha_{11}^2}{P_{B_1}}\right)+(\alpha_{10}+\alpha_{20}\sqrt{g_{12}}+\alpha_{11}-\mu_{1})^2}\right),\\
R_{10}+R_{11}+R_{20} &\leq& \frac{1}{2}\log\left(\frac{\mu_{1}^2}{\left(1+g_{12}P_{B_2}\right)\left(\frac{\alpha_{10}^2}{P_{A_1}}+\frac{\alpha_{20}^2}{P_{A_2}}+\frac{\alpha_{11}^2}{P_{B_1}}\right)+(\alpha_{10}+\alpha_{20}\sqrt{g_{12}}+\alpha_{11}-\mu_{1})^2}\right),\\
R_{20}+R_{22} &\leq& \frac{1}{2}\log\left(\frac{\left(1+P_{A_2}+P_{B_2}+g_{21} P_{B_1}\right)\frac{\alpha_{10}^2}{P_{A_1}}+(\alpha_{10}\sqrt{g_{21}}-\mu_{2})^2}{\left(1+g_{21}P_{B_1}\right)\left(\frac{\alpha_{20}^2}{P_{A_2}}+\frac{\alpha_{10}^2}{P_{A_1}}+\frac{\alpha_{22}^2}{P_{B_2}}\right)+(\alpha_{20}+\alpha_{10}\sqrt{g_{21}}+\alpha_{22}-\mu_{2})^2}\right),\\
R_{22}+R_{10} &\leq& \frac{1}{2}\log\left(\frac{\left(1+P_{B_2}+g_{21}P_{A_1}+g_{21} P_{B_1}\right)\frac{\alpha_{20}^2}{P_{A_2}}+(\alpha_{20}-\mu_{2})^2}{\left(1+g_{21}P_{B_1}\right)\left(\frac{\alpha_{20}^2}{P_{A_2}}+\frac{\alpha_{10}^2}{P_{A_1}}+\frac{\alpha_{22}^2}{P_{B_2}}\right)+(\alpha_{20}+\alpha_{10}\sqrt{g_{21}}+\alpha_{22}-\mu_{2})^2}\right),\\
R_{20}+R_{10} &\leq& \frac{1}{2}\log\left(\frac{\left(1+P_{A_2}+g_{21}P_{A_1}+g_{21} P_{B_1}\right)\frac{\alpha_{22}^2}{P_{B_2}}+(\alpha_{22}-\mu_{2})^2}{\left(1+g_{21}P_{B_1}\right)\left(\frac{\alpha_{20}^2}{P_{A_2}}+\frac{\alpha_{10}^2}{P_{A_1}}+\frac{\alpha_{22}^2}{P_{B_2}}\right)+(\alpha_{20}+\alpha_{10}\sqrt{g_{21}}+\alpha_{22}-\mu_{2})^2}\right),\\
R_{20}+R_{22}+R_{10} &\leq&
\frac{1}{2}\log\left(\frac{\mu_{2}^2}{\left(1+g_{21}P_{B_1}\right)\left(\frac{\alpha_{20}^2}{P_{A_2}}+\frac{\alpha_{10}^2}{P_{A_1}}+\frac{\alpha_{22}^2}{P_{B_2}}\right)+(\alpha_{20}+\alpha_{10}\sqrt{g_{21}}+\alpha_{22}-\mu_{2})^2}\right).
\end{eqnarray*}
As the state power $K\to\infty$, for any $(R_{10},R_{11},R_{20},R_{22}) \in
\widetilde{\mathcal{R}}_1'$, the rate pair $(R_{10}+R_{11},R_{20}+R_{22})$ is
achievable for the Gaussian IC with state information defined in Section
\ref{sec_2}.\label{lemma_gaussian_asy}
\end{Corollary}
\begin{Remark}
It can be easily seen that due to the special structure of
DPC~\cite{mac_dpc_laneman}, a nontrivial rate region can be achieved even when
the state power goes to infinity, as long as the state is non-causally known
at the transmitters.
\end{Remark}
In the following sections, we will consider several special cases of the
Gaussian IC with state information: the strong interference case, the mixed
interference case, and the weak interference case, respectively.
\section{The Strong Gaussian IC with State Information}\label{sec_5}
For the Gaussian IC with state information defined in Section \ref{sec_2}, the
channel is called strong Gaussian IC with state information if the
interference link gains satisfy $g_{21}\geq 1$ and $g_{12}\geq 1$. In this
section, we propose two achievable schemes for the strong Gaussian IC with
state information, and derive the corresponding achievable rate regions. An
enlarged achievable rate region is obtained by combining them with the
time-sharing technique.
\subsection{Scheme without Active Interference Cancellation}\label{sec_5_without_cancellation}
We first introduce a simple achievable scheme without active interference
cancellation, which is a building block towards the more general schemes coming
next. It is known that for the traditional strong Gaussian IC, the capacity
region can be obtained by the intersection of two MAC rate regions due to the
presence of the strong interference. However, for the strong Gaussian IC with
state information, the two MACs are not capacity-achieving simultaneously
since the optimal DPC parameters are different for these two MACs. Here we
propose a simple achievable scheme, which achieves the capacity for one of the
MACs and leaves the other MAC to suffer from the non-optimal DPC parameters.
Note that now all the source power is used to transmit the intended message at
both transmitters instead of being partly allocated to cancel the state effect
as in Section \ref{sec_4}.
\begin{Theorem}
Let $\mathcal{C}_{s1}$ be the set of all non-negative rate pairs $(R_1,R_2)$ satisfying
\begin{eqnarray*}
R_1 &\leq& \min\left\{\frac{1}{2}\log\left(1+P_1\right),\frac{1}{2}\log\left(\frac{\left(1+g_{21}P_{1}\right)\left(1+\frac{\alpha_{20}^2K}{P_2}\right)+K\left(\alpha_{20}-\frac{1}{\sqrt{N_2}}\right)^2}{1+\frac{\alpha_{20}^2K}{P_2}+\frac{\alpha_{10}^2K}{P_1}+K\left(\alpha_{20}+\alpha_{10}\sqrt{g_{21}}-\frac{1}{\sqrt{N_2}}\right)^2}\right)\right\},\\
R_2 &\leq&\frac{1}{2}\log\left(\frac{\left(1+P_{2}\right)\left(1+\frac{\alpha_{10}^2K}{P_1}\right)+K\left(\alpha_{10}\sqrt{g_{21}}-\frac{1}{\sqrt{N_2}}\right)^2}{1+\frac{\alpha_{20}^2K}{P_2}+\frac{\alpha_{10}^2K}{P_1}+K\left(\alpha_{20}+\alpha_{10}\sqrt{g_{21}}-\frac{1}{\sqrt{N_2}}\right)^2}\right),\\
R_1+R_2 &\leq& \min\left\{\frac{1}{2}\log\left(1+P_1+g_{12}P_2\right),\frac{1}{2}\log\left(\frac{1+P_{2}+g_{21}P_{1}+\frac{K}{N_2}}{1+\frac{\alpha_{20}^2K}{P_2}+\frac{\alpha_{10}^2K}{P_1}+K\left(\alpha_{20}+\alpha_{10}\sqrt{g_{21}}-\frac{1}{\sqrt{N_2}}\right)^2}\right)\right\},
\end{eqnarray*}
where $\alpha_{10}=\frac{P_1}{\sqrt{N_1}(1+P_1+g_{12}P_2)}$ and
$\alpha_{20}=\frac{\sqrt{g_{12}}P_2}{\sqrt{N_1}(1+P_1+g_{12}P_2)}$, which are
optimal for the MAC at receiver $1$. Then any rate pair $(R_1,R_2) \in
\mathcal{C}_{s1}$ is achievable for the strong Gaussian IC with state
information.

Similarly, let $\mathcal{C}_{s2}$ be the set of all non-negative rate pairs $(R_1,R_2)$ satisfying
\begin{eqnarray*}
R_1 &\leq& \frac{1}{2}\log\left(\frac{\left(1+P_{1}\right)\left(1+\frac{\alpha_{20}^2K}{P_2}\right)+K\left(\alpha_{20}\sqrt{g_{12}}-\frac{1}{\sqrt{N_1}}\right)^2}{1+\frac{\alpha_{10}^2K}{P_1}+\frac{\alpha_{20}^2K}{P_2}+K\left(\alpha_{10}+\alpha_{20}\sqrt{g_{12}}-\frac{1}{\sqrt{N_1}}\right)^2}\right),\\
R_2 &\leq& \min\left\{\frac{1}{2}\log\left(1+P_2\right),\frac{1}{2}\log\left(\frac{\left(1+g_{12}P_{2}\right)\left(1+\frac{\alpha_{10}^2K}{P_1}\right)+K\left(\alpha_{10}-\frac{1}{\sqrt{N_1}}\right)^2}{1+\frac{\alpha_{10}^2K}{P_1}+\frac{\alpha_{20}^2K}{P_2}+K\left(\alpha_{10}+\alpha_{20}\sqrt{g_{12}}-\frac{1}{\sqrt{N_1}}\right)^2}\right)\right\},\\
R_1+R_2 &\leq& \min\left\{\frac{1}{2}\log\left(1+P_2+g_{21}P_1\right),\frac{1}{2}\log\left(\frac{1+P_{1}+g_{12}P_{2}+ \frac{K}{N_1}}{1+\frac{\alpha_{10}^2K}{P_1}+\frac{\alpha_{20}^2K}{P_2}+K\left(\alpha_{10}+\alpha_{20}\sqrt{g_{12}}-\frac{1}{\sqrt{N_1}}\right)^2}\right)\right\},
\end{eqnarray*}
where $\alpha_{10}=\frac{\sqrt{g_{21}}P_1}{\sqrt{N_2}(1+P_2+g_{21}P_1)}$ and
$\alpha_{20}=\frac{P_2}{\sqrt{N_2}(1+P_2+g_{21}P_1)}$, which are optimal for
the MAC at receiver $2$). Then any rate pair $(R_1,R_2) \in \mathcal{C}_{s2}$
is achievable for the strong Gaussian IC with state
information.\label{theorem_strong_1}
\end{Theorem}
\begin{proof}
We only give the detailed proof for $\mathcal{C}_{s1}$ here. Similarly,
$\mathcal{C}_{s2}$ can be obtained by achieving the MAC capacity at receiver
$2$ and letting the MAC at receiver $1$ suffer from the non-optimal DPC
parameters.

Due to the presence of the strong interference, we only send common messages
at both transmitters instead of splitting the message into common and private
ones. Accordingly, both receivers need to decode the messages from both
transmitters. For the MAC at receiver $1$, the capacity region is given as:
\begin{eqnarray*}
R_1 &\leq& \frac{1}{2} \log\left(1+P_1\right) ,\\
R_2 &\leq& \frac{1}{2} \log\left(1+g_{12}P_2\right), \\
R_1+R_2 &\leq& \frac{1}{2} \log\left(1+P_1+g_{12}P_2\right),
\end{eqnarray*}
where DPC is utilized at both transmitters and the optimal DPC parameters are
$\alpha_{10}=\frac{P_1}{\sqrt{N_1}(1+P_1+g_{12}P_2)}$ and
$\alpha_{20}=\frac{\sqrt{g_{12}}P_2}{\sqrt{N_1}(1+P_1+g_{12}P_2)}$. However,
the MAC for receiver $2$ suffers from the non-optimal DPC parameters and has
the following achievable rate region:
\begin{eqnarray*}
R_1 &\leq& \frac{1}{2}\log\left(\frac{\left(1+g_{21}P_{1}\right)\left(1+\frac{\alpha_{20}^2K}{P_2}\right)+K\left(\alpha_{20}-\frac{1}{\sqrt{N_2}}\right)^2}{1+\frac{\alpha_{20}^2K}{P_2}+\frac{\alpha_{10}^2K}{P_1}+K\left(\alpha_{20}+\alpha_{10}\sqrt{g_{21}}-\frac{1}{\sqrt{N_2}}\right)^2}\right),\\
R_2 &\leq& \frac{1}{2}\log\left(\frac{\left(1+P_{2}\right)\left(1+\frac{\alpha_{10}^2K}{P_1}\right)+K\left(\alpha_{10}\sqrt{g_{21}}-\frac{1}{\sqrt{N_2}}\right)^2}{1+\frac{\alpha_{20}^2K}{P_2}+\frac{\alpha_{10}^2K}{P_1}+K\left(\alpha_{20}+\alpha_{10}\sqrt{g_{21}}-\frac{1}{\sqrt{N_2}}\right)^2}\right),\\
R_1+R_2 &\leq& \frac{1}{2}\log\left(\frac{1+P_{2}+g_{21}P_{1}+\frac{K}{N_2}}{1+\frac{\alpha_{20}^2K}{P_2}+\frac{\alpha_{10}^2K}{P_1}+K\left(\alpha_{20}+\alpha_{10}\sqrt{g_{21}}-\frac{1}{\sqrt{N_2}}\right)^2}\right).
\end{eqnarray*}
Consequently, we have the achievable region $\mathcal{C}_{s1}$ for the strong
Gaussian IC with state information, which is the intersection of the above two
rate regions for the two MACs.
\end{proof}
\subsection{Scheme with Active Interference Cancellation}\label{sec_5_with_cancellation}
For the strong Gaussian IC with state information, now we propose a more
general achievable scheme with active interference cancellation, which
allocates part of the source power to cancel the state effect at the
receivers. Specifically, DPC is used to achieve the capacity for one of the
MACs as shown in Section \ref{sec_5_without_cancellation}, and active
interference cancellation is employed at both transmitters to cancel the state
effect at the receivers. The corresponding achievable rate regions are
provided in the following theorem.
\begin{Theorem}
For any $\gamma_1^2 < P_1/K$ and $\gamma_2^2 < P_2/K$, let
$\mathcal{C}_{s3}(\gamma_1,\gamma_2)$ be the set of all non-negative rate
pairs $(R_1,R_2)$ satisfying
\setlength{\arraycolsep}{1pt}{\small\begin{eqnarray*}
R_1 &\leq& \min\left\{\frac{1}{2}\log\left(1+P_1-\gamma_1 ^2 K\right), \frac{1}{2}\log\left(\frac{\left(1+g_{21}(P_{1}-\gamma_1 ^2 K)\right)\left(1+\frac{\alpha_{20}^2K}{P_2-\gamma_2^2K}\right)+K\left(\alpha_{20}-\mu_{2}\right)^2}{1+\frac{\alpha_{20}^2K}{P_2-\gamma_2^2K}+\frac{\alpha_{10}^2K}{P_1-\gamma_1^2K}+K\left(\alpha_{20}+\alpha_{10}\sqrt{g_{21}}-\mu_{2}\right)^2}\right)\right\}\label{eq_strong_region_1_1},\\
R_2 &\leq& \frac{1}{2}\log\left(\frac{\left(1+P_{2}-\gamma_2^2K\right)\left(1+\frac{\alpha_{10}^2K}{P_1-\gamma_1^2K}\right)+K\left(\alpha_{10}\sqrt{g_{21}}-\mu_{2}\right)^2}{1+\frac{\alpha_{20}^2K}{P_2-\gamma_2^2K}+\frac{\alpha_{10}^2K}{P_1-\gamma_1^2K}+K\left(\alpha_{20}+\alpha_{10}\sqrt{g_{21}}-\mu_{2}\right)^2}\right),\label{eq_strong_region_1_2}\\
R_1+R_2 &\leq& \min\left\{\frac{1}{2}\log\left(1+P_1-\gamma_1 ^2 K+g_{12}\left(P_2-\gamma_2 ^2 K\right)\right),\frac{1}{2}\log\left(\frac{1+P_{2}-\gamma_2^2K+g_{21}(P_{1}-\gamma_1^2K)+\mu_{2}^2 K}{1+\frac{\alpha_{20}^2K}{P_2-\gamma_2^2K}+\frac{\alpha_{10}^2K}{P_1-\gamma_1^2K}+K\left(\alpha_{20}+\alpha_{10}\sqrt{g_{21}}-\mu_{2}\right)^2}\right)\right\},\label{eq_strong_region_1_3}
\end{eqnarray*}}
where
$\alpha_{10}=\frac{\mu_1(P_1-\gamma_1^2K)}{1+P_1-\gamma_1^2K+g_{12}(P_2-\gamma_2^2K)}$
and
$\alpha_{20}=\frac{\mu_1\sqrt{g_{12}}(P_2-\gamma_2^2K)}{1+P_1-\gamma_1^2K+g_{12}(P_2-\gamma_2^2K)}$,
which are optimal for the MAC at receiver $1$. Then any rate pair $(R_1,R_2)
\in \mathcal{C}_{s3}(\gamma_1,\gamma_2)$ is achievable for the strong Gaussian
IC with state information. Moreover, any rate pair in the convex hull (denoted
as $\hat{\mathcal{C}}_{s3}$) of $\mathcal{C}_{s3}(\gamma_1,\gamma_2)$ is also
achievable.

Similarly, for any $\gamma_1^2 < P_1/K$ and $\gamma_2^2 < P_2/K$, let
$\mathcal{C}_{s4}(\gamma_1,\gamma_2)$ be the set of all non-negative rate
pairs $(R_1,R_2)$ satisfying
\setlength{\arraycolsep}{1pt}{\small\begin{eqnarray*}
R_1 &\leq& \frac{1}{2}\log\left(\frac{\left(1+P_{1}-\gamma_1^2K\right)\left(1+\frac{\alpha_{20}^2K}{P_2-\gamma_2^2K}\right)+K\left(\alpha_{20}\sqrt{g_{12}}-\mu_1\right)^2}{1+\frac{\alpha_{10}^2K}{P_1-\gamma_1^2K}+\frac{\alpha_{20}^2K}{P_2-\gamma_2^2K}+K\left(\alpha_{10}+\alpha_{20}\sqrt{g_{12}}-\mu_1\right)^2}\right),\\
R_2 &\leq& \min\left\{\frac{1}{2}\log\left(1+P_2-\gamma_2^2K\right),\frac{1}{2}\log\left(\frac{\left(1+g_{12}(P_{2}-\gamma_2^2K)\right)\left(1+\frac{\alpha_{10}^2K}{P_1-\gamma_1^2K}\right)+K\left(\alpha_{10}-\mu_1\right)^2}{1+\frac{\alpha_{10}^2K}{P_1-\gamma_1^2K}+\frac{\alpha_{20}^2K}{P_2-\gamma_2^2K}+K\left(\alpha_{10}+\alpha_{20}\sqrt{g_{12}}-\mu_1\right)^2}\right)\right\},\\
R_1+R_2 &\leq& \min\left\{\frac{1}{2}\log\left(1+P_2-\gamma_2^2K+g_{21}(P_1-\gamma_1^2K)\right),\frac{1}{2}\log\left(\frac{1+P_{1}-\gamma_1^2K+g_{12}(P_{2}-\gamma_2^2K)+ K\mu_1^2}{1+\frac{\alpha_{10}^2K}{P_1-\gamma_1^2K}+\frac{\alpha_{20}^2K}{P_2-\gamma_2^2K}+K\left(\alpha_{10}+\alpha_{20}\sqrt{g_{12}}-\mu_1\right)^2}\right)\right\},
\end{eqnarray*}}
where
$\alpha_{10}=\frac{\mu_2\sqrt{g_{21}}(P_1-\gamma_1^2K)}{1+P_2-\gamma_2^2K+g_{21}(P_1-\gamma_1^2K)}$
and
$\alpha_{20}=\frac{\mu_2(P_2-\gamma_2^2K)}{1+P_2-\gamma_2^2K+g_{21}(P_1-\gamma_1^2K)}$,
which are optimal for the MAC at receiver $2$. Then any rate pair $(R_1,R_2)
\in \mathcal{C}_{s4}(\gamma_1,\gamma_2)$ is achievable for the strong Gaussian
IC with state information. Moreover, any rate pair in the convex hull (denoted
as $\hat{\mathcal{C}}_{s4}$) of $\mathcal{C}_{s4}(\gamma_1,\gamma_2)$ is also
achievable.\label{theorem_strong_2}
\end{Theorem}
The proof is omitted here since it is similar to that of Theorem
\ref{theorem_strong_1} except for applying active interference cancellation to
both users. Moreover, we see that the regions $\mathcal{C}_{s1}$ and
$\mathcal{C}_{s2}$ are equivalent to $\mathcal{C}_{s3}(0,0)$ and
$\mathcal{C}_{s4}(0,0)$, respectively, which means that the achievable scheme
without active interference cancellation is only a special case of the one with
active interference cancellation.

Note that an enlarged achievable rate region can be obtained by deploying the
time-sharing technique for any points in $\mathcal{C}_{s3}(\gamma_1,\gamma_2)$
and $\mathcal{C}_{s4}(\gamma_1,\gamma_2)$, which is described in the following
corollary.
\begin{Corollary}
The enlarged achievable rate region $\mathcal{C}_{s}$ for the strong Gaussian
IC with state information is given by the closure of the convex hull of
$\left(0,\frac{1}{2}\log\left(1+P_2\right)\right)$,
$\left(\frac{1}{2}\log\left(1+P_1\right),0\right)$, and all $(R_1,R_2)$ in
$\mathcal{C}_{s3}(\gamma_1,\gamma_2)$ and
$\mathcal{C}_{s4}(\gamma_1,\gamma_2)$ for any $\gamma_1^2<P_1/K$ and
$\gamma_2^2<P_2/K$.
\end{Corollary}
In Section~\ref{sec_9_strong}, we will numerically compare the above
achievable rate regions with an inner bound, which is denoted as
$\mathcal{C}_{s\_in}$ and defined by the achievable rate region when the
transmitters ignore the non-causal state information. The improvement due to
DPC and active interference cancellation is clearly shown there. We also
compare the above achievable rate regions with an outer bound (denoted by
$\mathcal{C}_{s\_o}$), which corresponds to the capacity region of the
traditional strong Gaussian IC~\cite{sato}. Such a correspondence is due to
the fact that the traditional Gaussian IC can be viewed as the idealization of
our channel model where the state is also known at the receivers.

\section{The Mixed Gaussian IC with State Information}\label{sec_6}
For the Gaussian IC with state information defined in Section \ref{sec_2}, the
channel is called mixed Gaussian IC with state information if the interference
link gains satisfy $g_{21}>1$, $g_{12}<1$ or $g_{21}<1$, $g_{12}>1$. In this
section, we propose two achievable schemes for the mixed Gaussian IC with
state information, and derive the corresponding achievable rate regions.
Similarly, we can enlarge the achievable rate region by combining them with
the time-sharing technique. Without loss of generality, from now on we assume
that $g_{21} > 1$ and $g_{12} < 1$.
\subsection{Scheme without Active Interference Cancellation}\label{sec_6_without_cancellation}
Similar to the strong Gaussian IC with state information, here we first
introduce a simple scheme without active interference cancellation, which
optimizes the DPC parameters for one receiver and leaves the other receiver
suffer from the non-optimal DPC parameters. Furthermore, receiver $1$ treats
the received signal from transmitter $2$ as noise, and receiver $2$ decodes
both messages from transmitter $1$ and transmitter $2$. Note that now all the
source power is used to send the intended messages at both transmitters
instead of employing active interference cancellation.
\begin{Theorem}
For any $\alpha_{22}$, let $\mathcal{C}_{m1}(\alpha_{22})$ be the set of all
non-negative rate pairs $(R_1,R_2)$ satisfying
\setlength{\arraycolsep}{1pt}{\small\begin{eqnarray*}
R_1 &\leq& \min\left\{\frac{1}{2}\log\left(1+\frac{P_1}{1+g_{12}P_2}\right),\frac{1}{2}\log\left(\frac{\left(1+g_{21}P_{1}\right)\left(1+\frac{\alpha_{22}^2K}{P_2}\right)+K\left(\alpha_{22}-\frac{1}{\sqrt{N_2}}\right)^2}{1+\frac{\alpha_{10}^2K}{P_1}+\frac{\alpha_{22}^2K}{P_2}+K\left(\alpha_{10}\sqrt{g_{21}}+\alpha_{22}-\frac{1}{\sqrt{N_2}}\right)^2}\right)\right\},\\
R_2 &\leq&\frac{1}{2}\log\left(\frac{\left(1+P_{2}\right)\left(1+\frac{\alpha_{10}^2K}{P_1}\right)+K\left(\alpha_{10}\sqrt{g_{21}}-\frac{1}{\sqrt{N_2}}\right)^2}{1+\frac{\alpha_{10}^2K}{P_1}+\frac{\alpha_{22}^2K}{P_2}+K\left(\alpha_{10}\sqrt{g_{21}}+\alpha_{22}-\frac{1}{\sqrt{N_2}}\right)^2}\right),\\
R_1+R_2 &\leq& \frac{1}{2}\log\left(\frac{1+P_{2}+g_{21}P_{1}+\frac{K}{N_2}}{1+\frac{\alpha_{10}^2K}{P_1}+\frac{\alpha_{22}^2K}{P_2}+K\left(\alpha_{10}\sqrt{g_{21}}+\alpha_{22}-\frac{1}{\sqrt{N_2}}\right)^2}\right),
\end{eqnarray*}}
where $\alpha_{10}=\frac{P_1}{\sqrt{N_1}(1+P_1+g_{12}P_2)}$ that is optimal
for the point-to-point link between transmitter $1$ and receiver $1$. Then any
rate pair $(R_1,R_2) \in \mathcal{C}_{m1}(\alpha_{22})$ is achievable for the
mixed Gaussian IC with state information. Moreover, any rate pair in the
convex hull (denoted as $\hat{\mathcal{C}}_{m1}$) of all
$\mathcal{C}_{m1}(\alpha_{22})$ is also achievable.

Similarly, let $\mathcal{C}_{m2}$ be the set of all non-negative rate pairs $(R_1,R_2)$ satisfying
\begin{eqnarray*}
R_1 &\leq& \frac{1}{2}\log\left(\frac{1+P_{1}+g_{12}P_2+\frac{K}{N_1}}{\left(1+g_{12}P_2\right)\left(1+\frac{\alpha_{10}^2K}{P_1}\right)+K\left(\alpha_{10}-\frac{1}{\sqrt{N_1}}\right)^2}\right),\\
R_2 &\leq& \frac{1}{2}\log\left(1+P_2\right),\\
R_1+R_2 &\leq& \frac{1}{2}\log\left(1+P_2+g_{21}P_1\right),
\end{eqnarray*}
where $\alpha_{10}=\frac{\sqrt{g_{21}}P_1}{\sqrt{N_2}(1+P_2+g_{21}P_1)}$ that
is optimal for the MAC at receiver $2$. Then any rate pair $(R_1,R_2) \in
\mathcal{C}_{m2}$ is achievable for the mixed Gaussian IC with state
information.\label{theorem_mixed_1}
\end{Theorem}
\begin{proof}
We only give the detailed derivation for $\mathcal{C}_{m1}$ here. The region
$\mathcal{C}_{m2}$ can be obtained in a similar manner by achieving the MAC
capacity at receiver $2$ and letting receiver $1$ suffer from the non-optimal
$\alpha_{10}$.

Since the interference link gains satisfy $g_{21}> 1$ and $g_{12} < 1$, the
interference for receiver $1$ is weaker than its intended signal and the
interference for receiver $2$ is stronger than its intended signal.
Accordingly, we send common message at transmitter $1$ and private message at
transmitter $2$ instead of splitting the message into common and private
messages for both transmitters. For the direct link from transmitter $1$ to
receiver $1$, the capacity is
\[
R_1 \leq \frac{1}{2}\log\left(1+\frac{P_1}{1+g_{12}P_2}\right),
\]
where the DPC parameter is $\alpha_{10} =
\frac{P_1}{\sqrt{N_1}(1+P_1+g_{12}P_2)}$. However, the MAC at receiver $2$
suffers from the non-optimal $\alpha_{10}$ and the achievable rate region is:
\begin{eqnarray*}
R_1 &\leq& \frac{1}{2}\log\left(\frac{\left(1+g_{21}P_{1}\right)\left(1+\frac{\alpha_{22}^2K}{P_2}\right)+K\left(\alpha_{22}-\frac{1}{\sqrt{N_2}}\right)^2}{1+\frac{\alpha_{10}^2K}{P_1}+\frac{\alpha_{22}^2K}{P_2}+K\left(\alpha_{10}\sqrt{g_{21}}+\alpha_{22}-\frac{1}{\sqrt{N_2}}\right)^2}\right),\\
R_2 &\leq&\frac{1}{2}\log\left(\frac{\left(1+P_{2}\right)\left(1+\frac{\alpha_{10}^2K}{P_1}\right)+K\left(\alpha_{10}\sqrt{g_{21}}-\frac{1}{\sqrt{N_2}}\right)^2}{1+\frac{\alpha_{10}^2K}{P_1}+\frac{\alpha_{22}^2K}{P_2}+K\left(\alpha_{10}\sqrt{g_{21}}+\alpha_{22}-\frac{1}{\sqrt{N_2}}\right)^2}\right),\\
R_1+R_2 &\leq& \frac{1}{2}\log\left(\frac{1+P_{2}+g_{21}P_{1}+\frac{K}{N_2}}{1+\frac{\alpha_{10}^2K}{P_1}+\frac{\alpha_{22}^2K}{P_2}+K\left(\alpha_{10}\sqrt{g_{21}}+\alpha_{22}-\frac{1}{\sqrt{N_2}}\right)^2}\right),
\end{eqnarray*}
for any $\alpha_{22}$. Therefore, we have the achievable rate region
$\mathcal{C}_{m1}(\alpha_{22})$ as the intersections of the above two regions.
\end{proof}
\subsection{Scheme with Active Interference Cancellation}\label{sec_6_with_cancellation}
Now we propose a more general scheme with active interference cancellation,
which allocates some source power to cancel the state effect at both
receivers. Similarly, the DPC parameters are only optimized for one receiver,
and the other receiver suffers from the non-optimal DPC parameters. The
corresponding achievable rate regions are stated in the following theorem.
\begin{Theorem}
For any $\alpha_{22}$, $\gamma_1^2 < P_1/K$, and $\gamma_2^2 < P_2/K$, let
$\mathcal{C}_{m3}(\alpha_{22},\gamma_1,\gamma_2)$ be the set of all
non-negative rate pairs $(R_1,R_2)$ satisfying
\setlength{\arraycolsep}{1pt}{\small\begin{eqnarray*}
R_1 &\leq& \min\left\{\frac{1}{2}\log\left(1+\frac{P_1-\gamma_1^2K}{1+g_{12}\left(P_2-\gamma_2^2K\right)}\right),\frac{1}{2}\log\left(\frac{\left(1+g_{21}\left(P_{1}-\gamma_1^2K\right)\right)\left(1+\frac{\alpha_{22}^2K}{P_2-\gamma_2^2K}\right)+K\left(\alpha_{22}-\mu_2\right)^2}{1+\frac{\alpha_{10}^2K}{P_1-\gamma_1^2K}+\frac{\alpha_{22}^2K}{P_2-\gamma_2^2K}+K\left(\alpha_{10}\sqrt{g_{21}}+\alpha_{22}-\mu_2\right)^2}\right)\right\},\\
R_2 &\leq&\frac{1}{2}\log\left(\frac{\left(1+P_{2}-\gamma_2^2K\right)\left(1+\frac{\alpha_{10}^2K}{P_1-\gamma_1^2K}\right)+K\left(\alpha_{10}\sqrt{g_{21}}-\mu_2\right)^2}{1+\frac{\alpha_{10}^2K}{P_1-\gamma_1^2K}+\frac{\alpha_{22}^2K}{P_2-\gamma_2^2K}+K\left(\alpha_{10}\sqrt{g_{21}}+\alpha_{22}-\mu_2\right)^2}\right),\\
R_1+R_2 &\leq& \frac{1}{2}\log\left(\frac{1+P_{2}-\gamma_2^2K+g_{21}\left(P_{1}-\gamma_1^2K\right)+K\mu_2^2}{1+\frac{\alpha_{10}^2K}{P_1-\gamma_1^2K}+\frac{\alpha_{22}^2K}{P_2-\gamma_2^2K}+K\left(\alpha_{10}\sqrt{g_{21}}+\alpha_{22}-\mu_2\right)^2}\right),
\end{eqnarray*}}
where
$\alpha_{10}=\frac{\mu_1(P_1-\gamma_1^2K)}{1+P_1-\gamma_1^2K+g_{12}(P_2-\gamma_2^2K)}$
that is optimal for the point-to-point link between transmitter $1$ and
receiver $1$. Then any rate pair $(R_1,R_2) \in
\mathcal{C}_{m3}(\alpha_{22},\gamma_1,\gamma_2)$ is achievable for the mixed
Gaussian IC with state information. Moreover, any rate pair in the convex hull
(denoted as $\hat{\mathcal{C}}_{m3}$) of
$\mathcal{C}_{m3}(\alpha_{22},\gamma_1,\gamma_2)$ is also achievable.

Similarly, for any $\gamma_1^2 < P_1/K$ and $\gamma_2^2 < P_2/K$, let
$\mathcal{C}_{m4}(\gamma_1,\gamma_2)$ be the set of all non-negative rate
pairs $(R_1,R_2)$ satisfying
\begin{eqnarray*}
R_1 &\leq& \frac{1}{2}\log\left(\frac{1+P_1-\gamma_1^2K+g_{12}\left(P_2-\gamma_2^2K\right)+K\mu_1^2}{\left(1+g_{12}\left(P_2-\gamma_2^2K\right)\right)\left(1+\frac{\alpha_{10}^2K}{P_1-\gamma_1^2K}\right)+K\left(\alpha_{10}-\mu_1\right)^2}\right),\\
R_2 &\leq& \frac{1}{2}\log\left(1+P_2-\gamma_2^2K\right),\\
R_1+R_2 &\leq& \frac{1}{2}\log\left(1+P_2-\gamma_2^2K+g_{21}\left(P_1-\gamma_1^2K\right)\right),
\end{eqnarray*}
where
$\alpha_{10}=\frac{\mu_2\sqrt{g_{21}}(P_1-\gamma_1^2K)}{1+P_2-\gamma_2^2K+g_{21}\left(P_1-\gamma_1^2K\right)}$
that is optimal for the MAC at receiver $2$). Then any rate pair $(R_1,R_2)
\in \mathcal{C}_{m4}(\gamma_1,\gamma_2)$ is achievable for the mixed Gaussian
IC with state information. Moreover, any rate pair in the convex hull (denoted
as $\hat{\mathcal{C}}_{m4}$) of $\mathcal{C}_{m4}(\gamma_1,\gamma_2)$ is also
achievable.\label{theorem_mixed_2}
\end{Theorem}
The proof is omitted here since it is similar to that of
Theorem~\ref{theorem_mixed_1} except for applying active interference
cancellation to both users. Moreover, it is straightforward to see that the
regions $\mathcal{C}_{m1}(\alpha_{22})$ and $\mathcal{C}_{m2}$ are equivalent
to $\mathcal{C}_{m3}(\alpha_{22},0,0)$ and $\mathcal{C}_{m4}(0,0)$,
respectively, which means that the achievable scheme without active
interference cancellation is only a special case of the one with active
interference cancellation.

Note that an enlarged achievable rate region can be obtained by deploying the
time-sharing technique for any points in
$\mathcal{C}_{m3}(\alpha_{22},\gamma_1,\gamma_2)$ and
$\mathcal{C}_{m4}(\gamma_1,\gamma_2)$, which is described in the following
corollary.
\begin{Corollary}
The enlarged achievable rate region $\mathcal{C}_{m}$ for the mixed Gaussian
IC with state information is given by the closure of the convex hull of
$\left(0,\frac{1}{2}\log\left(1+P_2\right)\right)$,
$\left(\frac{1}{2}\log\left(1+P_1\right),0\right)$, and all $(R_1,R_2)$ in
$\mathcal{C}_{m3}(\alpha_{22},\gamma_1,\gamma_2)$ and
$\mathcal{C}_{m4}(\gamma_1,\gamma_2)$ for any $\alpha_{22}$,
$\gamma_1^2<P_1/K$, and $\gamma_2^2<P_2/K$.
\end{Corollary}

In Section~\ref{sec_9_mixed}, we will numerically compare the above achievable
rate regions with an inner bound, which is denoted as $\mathcal{C}_{m\_in}$
and defined by the achievable rate region when the transmitters ignore the
non-causal state information. The improvement due to DPC and active
interference cancellation is clearly shown there. We also compare the above
achievable rate regions with an outer bound (denoted by $\mathcal{C}_{m\_o}$),
which is the outer bound derived for the traditional mixed Gaussian
IC~\cite{david}.

\subsection{A Special Case -- Degraded Gaussian IC}
For the Gaussian IC with state information defined in Section \ref{sec_2}, the
channel is called a degraded Gaussian IC with state information if the
interference link gains satisfy $g_{21}g_{12}=1$, which can be viewed as a
special case of the mixed Gaussian IC. For this degraded interference case, we
will show the numerical comparison between the achievable rate regions and the
outer bound in Section~\ref{sec_9_mixed}. Note that the difference from the
general mixed interference case is the evaluation of the outer bound
$\mathcal{C}_{m\_o}$, which is now equal to the outer bound including the sum
capacity for the traditional degraded Gaussian IC~\cite{sason}.
\section{The Weak Gaussian IC with State Information}\label{sec_7}
For the Gaussian IC with state information defined in Section \ref{sec_2}, the
channel is called weak Gaussian IC with state information if the interference
link gains satisfy $g_{21}<1$ and $g_{12}<1$. In this section, we propose
several achievable schemes for the weak Gaussian IC with state information,
and derive the corresponding achievable rate regions. An enlarged achievable
rate region is obtained by combining them with the time-sharing technique.
\subsection{Scheme without Active Interference Cancellation}\label{sec_7_without_cancellation}
We first introduce a simple scheme with fixed power allocation and without
active interference cancellation. It is shown in~\cite{david}
that for the traditional weak Gaussian IC, the achievable rate region is
within one bit of the capacity region if power splitting is chosen such that
the interfered private SNR at each receiver is equal to $1$. In our scheme, we
set the interfered private SNR equal to $1$, utilize sequential decoding, and
optimize the DPC parameters for one of the MACs. Note that now the power
allocation between the common message and private message is fixed, and all
the source power is used to transmit the intended message at both transmitters
instead of being partly allocated to cancel the state effect.
\begin{Theorem}
Let $\mathcal{C}_{w1}$ be the set of all non-negative rate pairs $(R_1,R_2)$
satisfying
\setlength{\arraycolsep}{1pt}{\footnotesize
\begin{eqnarray}
\nonumber R_1 &\leq& \min\left\{\frac{1}{2}\log\left(1+\frac{P_{A_1}}{1+P_{B_1}+g_{12}P_{B_2}}\right),\frac{1}{2}\log\left(\frac{\left(1+P_{B_2}+g_{21}P_{1}\right)\left(1+\frac{\alpha_{20}^2K}{P_{A_2}}\right)+K\left(\alpha_{20}-\frac{1}{\sqrt{N_2}}\right)^2}{\left(1+P_{B_2}+g_{21}P_{B_1}\right)\left(1+\frac{\alpha_{20}^2K}{P_{A_2}}+\frac{\alpha_{10}^2K}{P_{A_1}}\right)+K\left(\alpha_{20}+\alpha_{10}\sqrt{g_{21}}-\frac{1}{\sqrt{N_2}}\right)^2}\right)\right\}\\
\label{eq_Cw1_1} &&+\frac{1}{2}\log\left(1+\frac{P_{B_1}}{1+g_{12}P_{B_2}}\right),\\
\nonumber R_2 &\leq& \min\left\{\frac{1}{2}\log\left(1+\frac{g_{12}P_{A_2}}{1+P_{B_1}+g_{12}P_{B_2}}\right),\frac{1}{2}\log\left(\frac{\left(1+g_{21}P_{B_1}+P_{2}\right)\left(1+\frac{\alpha_{10}^2K}{P_{A_1}}\right)+K\left(\alpha_{10}\sqrt{g_{21}}-\frac{1}{\sqrt{N_2}}\right)^2}{\left(1+P_{B_2}+g_{21}P_{B_1}\right)\left(1+\frac{\alpha_{20}^2K}{P_{A_2}}+\frac{\alpha_{10}^2K}{P_{A_1}}\right)+K\left(\alpha_{20}+\alpha_{10}\sqrt{g_{21}}-\frac{1}{\sqrt{N_2}}\right)^2}\right)\right\}\\
\label{eq_Cw1_2} &&+\frac{1}{2}\log\left(1+\frac{P_{B_2}}{1+g_{21}P_{B_1}}\right),\\
\nonumber R_1+R_2 &\leq& \min\left\{\frac{1}{2}\log\left(1+\frac{P_{A_1}+g_{12}P_{A_2}}{1+P_{B_1}+g_{12}P_{B_2}}\right),\frac{1}{2}\log\left(\frac{1+P_{2}+g_{21}P_{1}+\frac{K}{N_2}}{\left(1+P_{B_2}+g_{21}P_{B_1}\right)\left(1+\frac{\alpha_{20}^2K}{P_{A_2}}+\frac{\alpha_{10}^2K}{P_{A_1}}\right)+K\left(\alpha_{20}+\alpha_{10}\sqrt{g_{21}}-\frac{1}{\sqrt{N_2}}\right)^2}\right)\right\}\\
\label{eq_Cw1_3} &&+\frac{1}{2}\log\left(1+\frac{P_{B_1}}{1+g_{12}P_{B_2}}\right)+\frac{1}{2}\log\left(1+\frac{P_{B_2}}{1+g_{21}P_{B_1}}\right),
\end{eqnarray}}
where $P_{B_1}=\min\{P_1,1/g_{21}\}$, $P_{B_2}=\min\{P_2,1/g_{12}\}$,
$\alpha_{10}=\frac{P_{A_1}}{\sqrt{N_1}(1+P_1+g_{12}P_2)}$, and
$\alpha_{20}=\frac{\sqrt{g_{12}}P_{A_2}}{\sqrt{N_1}(1+P_1+g_{12}P_2)}$. Then
any rate pair $(R_1,R_2) \in \mathcal{C}_{w1}$ is achievable for the weak
Gaussian IC with state information.

Similarly, let $\mathcal{C}_{w2}$ be the set of all non-negative rate pairs
$(R_1,R_2)$ satisfying
\setlength{\arraycolsep}{1pt}{\footnotesize
\begin{eqnarray}
\nonumber R_1 &\leq& \min\left\{\frac{1}{2}\log\left(1+\frac{g_{21}P_{A_1}}{1+P_{B_2}+g_{21}P_{B_1}}\right),\frac{1}{2}\log\left(\frac{\left(1+g_{12}P_{B_2}+P_{1}\right)\left(1+\frac{\alpha_{20}^2K}{P_{A_2}}\right)+K\left(\alpha_{20}\sqrt{g_{12}}-\frac{1}{\sqrt{N_1}}\right)^2}{\left(1+P_{B_1}+g_{12}P_{B_2}\right)\left(1+\frac{\alpha_{10}^2K}{P_{A_1}}+\frac{\alpha_{20}^2K}{P_{A_2}}\right)+K\left(\alpha_{10}+\alpha_{20}\sqrt{g_{12}}-\frac{1}{\sqrt{N_1}}\right)^2}\right)\right\}\\
\label{eq_Cw2_1} &&+\frac{1}{2}\log\left(1+\frac{P_{B_1}}{1+g_{12}P_{B_2}}\right),\\
\nonumber R_2 &\leq& \min\left\{\frac{1}{2}\log\left(1+\frac{P_{A_2}}{1+P_{B_2}+g_{21}P_{B_1}}\right),\frac{1}{2}\log\left(\frac{\left(1+P_{B_1}+g_{12}P_{2}\right)\left(1+\frac{\alpha_{10}^2K}{P_{A_1}}\right)+K\left(\alpha_{10}-\frac{1}{\sqrt{N_1}}\right)^2}{\left(1+P_{B_1}+g_{12}P_{B_2}\right)\left(1+\frac{\alpha_{10}^2K}{P_{A_1}}+\frac{\alpha_{20}^2K}{P_{A_2}}\right)+K\left(\alpha_{10}+\alpha_{20}\sqrt{g_{12}}-\frac{1}{\sqrt{N_1}}\right)^2}\right)\right\}\\
\label{eq_Cw2_2} &&+\frac{1}{2}\log\left(1+\frac{P_{B_2}}{1+g_{21}P_{B_1}}\right),\\
\nonumber R_1+R_2 &\leq& \min\left\{\frac{1}{2}\log\left(1+\frac{P_{A_2}+g_{21}P_{A_1}}{1+P_{B_2}+g_{21}P_{B_1}}\right),\frac{1}{2}\log\left(\frac{1+P_{1}+g_{12}P_{2}+\frac{K}{N_1}}{\left(1+P_{B_1}+g_{12}P_{B_2}\right)\left(1+\frac{\alpha_{10}^2K}{P_{A_1}}+\frac{\alpha_{20}^2K}{P_{A_2}}\right)+K\left(\alpha_{10}+\alpha_{20}\sqrt{g_{12}}-\frac{1}{\sqrt{N_1}}\right)^2}\right)\right\}\\
\label{eq_Cw2_3} &&+\frac{1}{2}\log\left(1+\frac{P_{B_2}}{1+g_{21}P_{B_1}}\right)+\frac{1}{2}\log\left(1+\frac{P_{B_1}}{1+g_{12}P_{B_2}}\right),
\end{eqnarray}}
where $P_{B_1}=\min\{P_1,1/g_{21}\}$, $P_{B_2}=\min\{P_2,1/g_{12}\}$,
$\alpha_{10}=\frac{\sqrt{g_{21}}P_{A_1}}{\sqrt{N_2}(1+P_2+g_{21}P_1)}$, and
$\alpha_{20}=\frac{P_{A_2}}{\sqrt{N_2}(1+P_2+g_{21}P_1)}$. Then any rate pair
$(R_1,R_2) \in \mathcal{C}_{w2}$ is achievable for the weak Gaussian IC with
state information.\label{theorem_weak_1}
\end{Theorem}
\begin{proof}
We only give the detailed proof for $\mathcal{C}_{w1}$ here. Similarly,
$\mathcal{C}_{w2}$ can be obtained by optimizing the DPC parameters for the
common messages at receiver $2$ and letting the common-message MAC at receiver
$1$ suffer from the non-optimal DPC parameters.

Due to the presence of the weak interference, we split the message into common
and private ones at both transmitters. The sequential decoder is utilized at
the receivers, i.e., both receivers first decode both common messages by
treating both private messages as noise, and then decode the intended private
message by treating the interfered private message as noise. For the
common-message MAC at receiver $1$, the capacity region is given as follows:
\begin{eqnarray*}
R_{10} &\leq& \frac{1}{2}\log\left(1+\frac{P_{A_1}}{1+P_{B_1}+g_{12}P_{B_2}}\right),\\
R_{20} &\leq& \frac{1}{2}\log\left(1+\frac{g_{12}P_{A_2}}{1+P_{B_1}+g_{12}P_{B_2}}\right),\\
R_{10}+R_{20} &\leq&
\frac{1}{2}\log\left(1+\frac{P_{A_1}+g_{12}P_{A_2}}{1+P_{B_1}+g_{12}P_{B_2}}\right),
\end{eqnarray*}
where $P_{B_1}=\min\{P_1,1/g_{21}\}$, $P_{B_2}=\min\{P_2,1/g_{12}\}$, and DPC
is utilized for both common messages with the optimal DPC parameters
$\alpha_{10}=\frac{P_{A_1}}{\sqrt{N_1}(1+P_1+g_{12}P_2)}$ and
$\alpha_{20}=\frac{\sqrt{g_{12}}P_{A_2}}{\sqrt{N_1}(1+P_1+g_{12}P_2)}$.
However, the common-message MAC at receiver $2$ suffers from the non-optimal
DPC parameters and has the following achievable rate region:
\begin{eqnarray*}
R_{10} &\leq& \frac{1}{2}\log\left(\frac{\left(1+P_{B_2}+g_{21}P_{1}\right)\left(1+\frac{\alpha_{20}^2K}{P_{A_2}}\right)+K\left(\alpha_{20}-\frac{1}{\sqrt{N_2}}\right)^2}{\left(1+P_{B_2}+g_{21}P_{B_1}\right)\left(1+\frac{\alpha_{20}^2K}{P_{A_2}}+\frac{\alpha_{10}^2K}{P_{A_1}}\right)+K\left(\alpha_{20}+\alpha_{10}\sqrt{g_{21}}-\frac{1}{\sqrt{N_2}}\right)^2}\right),\\
R_{20} &\leq& \frac{1}{2}\log\left(\frac{\left(1+g_{21}P_{B_1}+P_{2}\right)\left(1+\frac{\alpha_{10}^2K}{P_{A_1}}\right)+K\left(\alpha_{10}\sqrt{g_{21}}-\frac{1}{\sqrt{N_2}}\right)^2}{\left(1+P_{B_2}+g_{21}P_{B_1}\right)\left(1+\frac{\alpha_{20}^2K}{P_{A_2}}+\frac{\alpha_{10}^2K}{P_{A_1}}\right)+K\left(\alpha_{20}+\alpha_{10}\sqrt{g_{21}}-\frac{1}{\sqrt{N_2}}\right)^2}\right),\\
R_{10}+R_{20} &\leq&
\frac{1}{2}\log\left(\frac{1+P_{2}+g_{21}P_{1}+\frac{K}{N_2}}{\left(1+P_{B_2}+g_{21}P_{B_1}\right)\left(1+\frac{\alpha_{20}^2K}{P_{A_2}}+\frac{\alpha_{10}^2K}{P_{A_1}}\right)+K\left(\alpha_{20}+\alpha_{10}\sqrt{g_{21}}-\frac{1}{\sqrt{N_2}}\right)^2}\right).
\end{eqnarray*}

Consequently, the IC achievable region for the common messages can be obtained
by intersecting the above regions for the two MACs. After decoding the common
messages, each receiver is capable of decoding the intended private message
with the following rate:
\begin{eqnarray*}
R_{11} &\leq& \frac{1}{2}\log\left(1+\frac{P_{B_1}}{1+g_{12}P_{B_2}}\right),\\
R_{22} &\leq& \frac{1}{2}\log\left(1+\frac{P_{B_2}}{1+g_{21}P_{B_1}}\right).
\end{eqnarray*}

Therefore, after applying the Fourier-Motzkin algorithm, we have the
achievable region $\mathcal{C}_{w1}$ for the weak Gaussian IC with state
information.
\end{proof}
\subsection{Scheme with Active Interference Cancellation}\label{sec_7_with_cancellation}
For the weak Gaussian IC with state information, now we generalize the
previous scheme with active interference cancellation, which allocates part of
the source power to cancel the state effect at the receivers. Specifically,
DPC is used to achieve the capacity for one of the common-message MACs as
shown in Section \ref{sec_5_without_cancellation}, and active interference
cancellation is deployed to cancel the state effect at the receivers. The
corresponding achievable rate regions are provided in the following theorem.
\begin{Theorem}
For any $\gamma_1^2 < P_{A1}/K$ and $\gamma_2^2 < P_{A2}/K$, let
$\mathcal{C}_{w3}(\gamma_1,\gamma_2)$ be the set of all non-negative rate
pairs $(R_1,R_2)$ satisfying
\setlength{\arraycolsep}{1pt}{\tiny
\begin{eqnarray*}
R_1 &\leq& \min\left\{\frac{1}{2}\log\left(1+\frac{P_{A_1}-\gamma_1^2 K}{1+P_{B_1}+g_{12}P_{B_2}}\right),\frac{1}{2}\log\left(\frac{\left(1+P_{B_2}+g_{21}(P_{1}-\gamma_1^2 K)\right)\left(1+\frac{\alpha_{20}^2K}{P_{A_2}-\gamma_2^2 K}\right)+K\left(\alpha_{20}-\mu_2\right)^2}{\left(1+P_{B_2}+g_{21}P_{B_1}\right)\left(1+\frac{\alpha_{20}^2K}{P_{A_2}-\gamma_2^2 K}+\frac{\alpha_{10}^2K}{P_{A_1}-\gamma_1^2 K}\right)+K\left(\alpha_{20}+\alpha_{10}\sqrt{g_{21}}-\mu_2\right)^2}\right)\right\}\\
&&+\frac{1}{2}\log\left(1+\frac{P_{B_1}}{1+g_{12}P_{B_2}}\right),\\
R_2 &\leq& \min\left\{\frac{1}{2}\log\left(1+\frac{g_{12}(P_{A_2}-\gamma_2^2 K)}{1+P_{B_1}+g_{12}P_{B_2}}\right),\frac{1}{2}\log\left(\frac{\left(1+g_{21}P_{B_1}+P_{2}-\gamma_2^2 K\right)\left(1+\frac{\alpha_{10}^2K}{P_{A_1}-\gamma_1^2 K}\right)+K\left(\alpha_{10}\sqrt{g_{21}}-\mu_2\right)^2}{\left(1+P_{B_2}+g_{21}P_{B_1}\right)\left(1+\frac{\alpha_{20}^2K}{P_{A_2}-\gamma_2^2 K}+\frac{\alpha_{10}^2K}{P_{A_1}-\gamma_1^2 K}\right)+K\left(\alpha_{20}+\alpha_{10}\sqrt{g_{21}}-\mu_2\right)^2}\right)\right\}\\
&&+\frac{1}{2}\log\left(1+\frac{P_{B_2}}{1+g_{21}P_{B_1}}\right),\\
R_1+R_2 &\leq& \min\left\{\frac{1}{2}\log\left(1+\frac{P_{A_1}-\gamma_1^2 K+g_{12}(P_{A_2}-\gamma_2^2 K)}{1+P_{B_1}+g_{12}P_{B_2}}\right),\frac{1}{2}\log\left(\frac{1+P_{2}-\gamma_2^2 K+g_{21}(P_{1}-\gamma_1^2 K)+\mu_2^2 K N_2}{\left(1+P_{B_2}+g_{21}P_{B_1}\right)\left(1+\frac{\alpha_{20}^2K}{P_{A_2}-\gamma_2^2 K}+\frac{\alpha_{10}^2K}{P_{A_1}-\gamma_1^2 K}\right)+K\left(\alpha_{20}+\alpha_{10}\sqrt{g_{21}}-\mu_2\right)^2}\right)\right\}\\
&&+\frac{1}{2}\log\left(1+\frac{P_{B_1}}{1+g_{12}P_{B_2}}\right)+\frac{1}{2}\log\left(1+\frac{P_{B_2}}{1+g_{21}P_{B_1}}\right),
\end{eqnarray*}}
where $P_{B_1}=\min\{P_1,1/g_{21}\}$, $P_{B_2}=\min\{P_2,1/g_{12}\}$,
$\alpha_{10}=\frac{\mu_1(P_{A_1}-\gamma_1^2 K)}{\left(1+P_1-\gamma_1^2
K+g_{12}\left(P_2-\gamma_2^2 K\right)\right)}$, and
$\alpha_{20}=\frac{\mu_1\sqrt{g_{12}}(P_{A_2}-\gamma_2^2 K)}{1+P_1-\gamma_1^2
K+g_{12}(P_2-\gamma_2^2 K)}$, which are optimal for the common-message MAC at
receiver $1$. Then any rate pair $(R_1,R_2) \in
\mathcal{C}_{w3}(\gamma_1,\gamma_2)$ is achievable for the weak Gaussian IC
with state information. Moreover, any rate pair in the convex hull (denoted as
$\hat{\mathcal{C}}_{w3}$) of $\mathcal{C}_{w3}(\gamma_1,\gamma_2)$ is also
achievable.

Similarly, for any $\gamma_1^2 < P_{A1}/K$ and $\gamma_2^2 < P_{A2}/K$, let
$\mathcal{C}_{w4}(\gamma_1,\gamma_2)$ be the set of all non-negative rate
pairs $(R_1,R_2)$ satisfying
\setlength{\arraycolsep}{1pt}{\tiny
\begin{eqnarray*}
R_1 &\leq& \min\left\{\frac{1}{2}\log\left(1+\frac{g_{21}(P_{A_1}-\gamma_1^2 K)}{1+P_{B_2}+g_{21}P_{B_1}}\right),\frac{1}{2}\log\left(\frac{\left(1+g_{12}P_{B_2}+P_{1}-\gamma_1^2 K\right)\left(1+\frac{\alpha_{20}^2K}{P_{A_2}-\gamma_2^2 K}\right)+K\left(\alpha_{20}\sqrt{g_{12}}-\mu_1\right)^2}{\left(1+P_{B_1}+g_{12}P_{B_2}\right)\left(1+\frac{\alpha_{10}^2K}{P_{A_1}-\gamma_1^2 K}+\frac{\alpha_{20}^2K}{P_{A_2}-\gamma_2^2 K}\right)+K\left(\alpha_{10}+\alpha_{20}\sqrt{g_{12}}-\mu_1\right)^2}\right)\right\}\\
&&+\frac{1}{2}\log\left(1+\frac{P_{B_1}}{1+g_{12}P_{B_2}}\right),\\
R_2 &\leq& \min\left\{\frac{1}{2}\log\left(1+\frac{P_{A_2}-\gamma_2^2 K}{1+P_{B_2}+g_{21}P_{B_1}}\right),\frac{1}{2}\log\left(\frac{\left(1+P_{B_1}+g_{12}(P_{2}-\gamma_2^2 K)\right)\left(1+\frac{\alpha_{10}^2K}{P_{A_1}-\gamma_1^2 K}\right)+K\left(\alpha_{10}-\mu_1\right)^2}{\left(1+P_{B_1}+g_{12}P_{B_2}\right)\left(1+\frac{\alpha_{10}^2K}{P_{A_1}-\gamma_1^2 K}+\frac{\alpha_{20}^2K}{P_{A_2}-\gamma_2^2 K}\right)+K\left(\alpha_{10}+\alpha_{20}\sqrt{g_{12}}-\mu_1\right)^2}\right)\right\}\\
&&+\frac{1}{2}\log\left(1+\frac{P_{B_2}}{1+g_{21}P_{B_1}}\right),\\
R_1+R_2 &\leq& \min\left\{\frac{1}{2}\log\left(1+\frac{P_{A_2}-\gamma_2^2 K+g_{21}(P_{A_1}-\gamma_1^2 K)}{1+P_{B_2}+g_{21}P_{B_1}}\right),\frac{1}{2}\log\left(\frac{1+P_{1}-\gamma_1^2 K+g_{12}(P_{2}-\gamma_2^2 K)+\mu_1^2 K}{\left(1+P_{B_1}+g_{12}P_{B_2}\right)\left(1+\frac{\alpha_{10}^2K}{P_{A_1}-\gamma_1^2 K}+\frac{\alpha_{20}^2K}{P_{A_2}-\gamma_2^2 K}\right)+K\left(\alpha_{10}+\alpha_{20}\sqrt{g_{12}}-\mu_1\right)^2}\right)\right\}\\
&&+\frac{1}{2}\log\left(1+\frac{P_{B_2}}{1+g_{21}P_{B_1}}\right)+\frac{1}{2}\log\left(1+\frac{P_{B_1}}{1+g_{12}P_{B_2}}\right),
\end{eqnarray*}}
where $P_{B_1}=\min\{P_1,1/g_{21}\}$, $P_{B_2}=\min\{P_2,1/g_{12}\}$,
$\alpha_{10}=\frac{\mu_2\sqrt{g_{21}}(P_{A_1}-\gamma_1^2 K)}{1+P_2-\gamma_2^2
K+g_{21}(P_1-\gamma_1^2 K)}$, and $\alpha_{20}=\frac{\mu_2(P_{A_2}-\gamma_2^2
K)}{1+P_2-\gamma_2^2 K+g_{21}(P_1-\gamma_1^2 K)}$, which are optimal for the
common-message MAC at receiver $2$. Then any rate pair $(R_1,R_2) \in
\mathcal{C}_{w4}(\gamma_1,\gamma_2)$ is achievable for the weak Gaussian IC
with state information. Moreover, any rate pair in the convex hull (denoted as
$\hat{\mathcal{C}}_{w4}$) of $\mathcal{C}_{w4}(\gamma_1,\gamma_2)$ is also
achievable.\label{theorem_weak_2}
\end{Theorem}
The proof is omitted here since it is similar to that of Theorem
\ref{theorem_weak_1} except for applying active interference cancellation to
both users. Moreover, we see that the regions $\mathcal{C}_{w1}$ and
$\mathcal{C}_{w2}$ are equivalent to $\mathcal{C}_{w3}(0,0)$ and
$\mathcal{C}_{w4}(0,0)$, respectively, which again implies that the achievable
scheme without active interference cancellation is only a special case of the
one with active interference cancellation.

As in previous sections, an enlarged achievable rate region can be obtained by
employing the time-sharing technique for any points in
$\mathcal{C}_{w3}(\gamma_1,\gamma_2)$ and
$\mathcal{C}_{w4}(\gamma_1,\gamma_2)$, which is described in the following
corollary.
\begin{Corollary}
The enlarged achievable rate region $\mathcal{C}_{w}$ for the weak Gaussian IC
with state information is given by the closure of the convex hull of
$\left(0,\frac{1}{2}\log\left(1+P_2\right)\right)$,
$\left(\frac{1}{2}\log\left(1+P_1\right),0\right)$, and all $(R_1,R_2)$ in
$\mathcal{C}_{w3}(\gamma_1,\gamma_2)$ and
$\mathcal{C}_{w4}(\gamma_1,\gamma_2)$ for any $\gamma_1^2 < P_{A1}/K$ and
$\gamma_2^2 < P_{A2}/K$.
\end{Corollary}

In Section~\ref{sec_9_weak}, we will numerically compare the above achievable
rate regions with an inner bound, which is denoted as $\mathcal{C}_{w\_in}$
and defined by the achievable rate region when the transmitters ignore the
non-causal state information. We also compare the above achievable rate
regions with an outer bound (denoted by $\mathcal{C}_{w\_o}$), which is the
outer bound derived for the traditional weak Gaussian IC~\cite{david}. Note
that unlike the strong interference case and the mixed interference case,
active interference cancellation cannot enlarge the achievable rate region
significantly for the weak interference case. Intuitively, the reason is that
the source power is too ``precious" to cancel the state effect when the
interference is weak. Therefore, we next modify the scheme to optimize the
power allocation between the common message and the private message at each
transmitter.
\subsection{Scheme with Flexible Power Allocation}\label{sec_7_with_power_allocation}
For the weak Gaussian IC with state information, now we propose a scheme with
flexible power allocation. The corresponding achievable rate regions are
provided in the following theorem.
\begin{Theorem}
For any $\beta_1, \beta_2 \in (0,1)$, let $\mathcal{C}_{w5}(\beta_1,\beta_2)$
be the set of all non-negative rate pairs $(R_1,R_2)$ satisfying
\eqref{eq_Cw1_1}-\eqref{eq_Cw1_3} where $P_{B_1}=\beta_1 P_1$,
$P_{B_2}=\beta_2 P_2$,
$\alpha_{10}=\frac{(1-\beta_1)P_1}{\sqrt{N_1}\left(1+P_1+g_{12}P_2\right)}$,
and
$\alpha_{20}=\frac{\sqrt{g_{12}}(1-\beta_2)P_{2}}{\sqrt{N_1}\left(1+P_1+g_{12}P_2\right)}$,
which are optimal for the common-message MAC at receiver $1$. Then any rate
pair $(R_1,R_2) \in \mathcal{C}_{w5}(\beta_1,\beta_2)$ is achievable for the
weak Gaussian IC with state information. Moreover, any rate pair in the convex
hull (denoted as $\hat{\mathcal{C}}_{w5}$) of
$\mathcal{C}_{w5}(\beta_1,\beta_2)$ is also achievable.

Similarly, for any $\beta_1, \beta_2 \in (0,1)$, let
$\mathcal{C}_{w6}(\beta_1,\beta_2)$ be the set of all non-negative rate pairs
$(R_1,R_2)$ satisfying \eqref{eq_Cw2_1}-\eqref{eq_Cw2_3}, where
$P_{B_1}=\beta_1 P_1$, $P_{B_2}=\beta_2 P_2$,
$\alpha_{10}=\frac{\sqrt{g_{21}}(1-\beta_1)P_{1}}{\sqrt{N_2}\left(1+P_2+g_{21}P_1\right)}$,
and
$\alpha_{20}=\frac{(1-\beta_2)P_{2}}{\sqrt{N_2}\left(1+P_2+g_{21}\right)}$,
which are optimal for the common-message MAC at receiver $2$. Then any rate
pair $(R_1,R_2) \in \mathcal{C}_{w6}(\beta_1,\beta_2)$ is achievable for the
weak Gaussian IC with state information. Moreover, any rate pair in the convex
hull (denoted as $\hat{\mathcal{C}}_{w6}$) of
$\mathcal{C}_{w6}(\beta_1,\beta_2)$ is also achievable.\label{theorem_weak_3}
\end{Theorem}
The proof is omitted here since it is similar to that of Theorem
\ref{theorem_weak_1} except for applying the optimal power allocation between
the common and private messages at both transmitters, which is obtained by
two-dimensional searching and bears the same complexity as the active
interference cancellation scheme in Section~\ref{sec_7_with_cancellation}.
Similarly, an enlarged achievable rate region can be obtained by employing the
time-sharing technique for any points in $\mathcal{C}_{w5}(\beta_1,\beta_2)$
and $\mathcal{C}_{w6}(\beta_1,\beta_2)$, which is described in the following
corollary.
\begin{Corollary}
The enlarged achievable rate region $\hat{\mathcal{C}}_{w}$ for the weak
Gaussian IC with state information is given by the closure of the convex hull
of $\left(0,\frac{1}{2}\log\left(1+P_2\right)\right)$,
$\left(\frac{1}{2}\log\left(1+P_1\right),0\right)$, and all $(R_1,R_2)$ in
$\mathcal{C}_{w5}(\beta_1,\beta_2)$ and $\mathcal{C}_{w6}(\beta_1,\beta_2)$
for any $\beta_1, \beta_2 \in (0,1)$.
\end{Corollary}

The numerical comparison between the above achievable rate regions with the
outer bound $\mathcal{C}_{w\_o}$~\cite{david} is shown in
Section~\ref{sec_9_weak}.

\subsection{Scheme with Flexible Sequential Decoder}\label{sec_7_with_different_sequential decoder}
For the sequential decoder of $\mathcal{C}_{w1}$ in
Section~\ref{sec_7_without_cancellation}, each receiver first decodes the
common messages by treating the private messages as noise, then decodes the
intended private message by treating the interfered private message as noise.
Note that we can easily extend the above scheme by changing the decoding
order. For example, receiver $1$ could also decode the intended common message
and private message first, or decode the ``interfered" common message and
intended private message first. Therefore, each receiver has $3$ choices of
different sequential decoders, which means that there are $9$ different
choices with two receivers. Similarly, we could have another $9$ choices based
on the sequential decoder of $\mathcal{C}_{w2}$, which optimizes the DPC
parameter at the MAC for receiver $2$. Finally, we can apply Fourier-Motzkin
algorithm for each implicit achievable rate region corresponding to each
decoder ($18$ different decoders in total), then obtain the explicit
achievable rate regions, and finally deploy the time-sharing technique to
enlarge the achievable rate region. The details are omitted here due to its
similarity to the previous results.

\section{Numerical Results}\label{sec_9}

In this section, we compare the derived various achievable rate regions with
the outer bound, which is the same as the outer bound derived for the
traditional Gaussian IC~\cite{sato,sason,david}, since the traditional IC can
be treated as the idealization of our model where the state is also known at
the receivers. We show the numerical results for three cases: the strong
interference case, the mixed interference case, and the weak interference
case. From the numerical comparison, we can easily see that active
interference cancellation significantly enlarges the achievable rate region for
the strong and mixed interference case. However, for the weak interference
case, flexible power allocation brings more benefit due to the ``preciousness"
of the transmission power.

\subsection{Strong Gaussian IC with State Information}\label{sec_9_strong}
In Fig. \ref{fig_strong_ic_1}, we compare the achievable rate regions in
Section~\ref{sec_5_with_cancellation} with the outer bound
$\mathcal{C}_{s\_o}$, which is the capacity region of the traditional strong
Gaussian IC with the state information also known at the
receivers~\cite{sato}. Note that the inner bound $\mathcal{C}_{s\_in}$ is
defined as the rate region when the transmitters ignore the non-causal state
information. Compared with $\mathcal{C}_{s1}$ and $\mathcal{C}_{s2}$ (only
utilizing DPC), we see that the knowledge of the state information at the
transmitters improves the performance significantly by deploying DPC.
Moreover, it can be easily seen that $\hat{\mathcal{C}}_{s3}$ and
$\hat{\mathcal{C}}_{s4}$ (utilizing DPC and active interference cancellation)
are much bigger than $\mathcal{C}_{s1}$ and $\mathcal{C}_{s2}$, respectively,
which implies that active interference cancellation enlarges the achievable
rate region significantly. Finally, we observe that the achievable rate region
$\mathcal{C}_{s}$ is fairly close to the outer bound, even when the state
power is the same as the source power.
\begin{figure}[!t]
\centering
\includegraphics[width=3.5in]{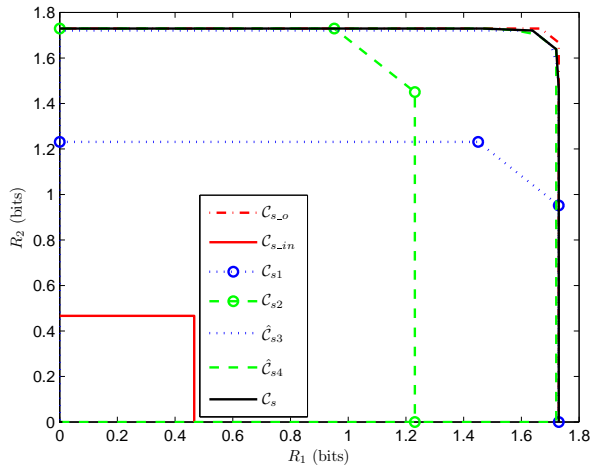}
\caption{Comparison of different achievable rate regions and the outer bound
for the strong Gaussian IC with state information. The channel parameters are
set as: $g_{12} = g_{21} = 10$, $N_1=N_2=1$, $P_1=P_2=K=10\textrm{ dB}$.}
\label{fig_strong_ic_1}
\end{figure}

\subsection{Mixed Gaussian IC with State Information}\label{sec_9_mixed}
In Fig. \ref{fig_mixed_ic_1}, we compare the achievable rate regions in
Section~\ref{sec_6_with_cancellation} with the outer bound
$\mathcal{C}_{m\_o}$, which is the same as the outer bound derived for the
traditional mixed Gaussian IC~\cite{david}. Also we define the inner bound
$\mathcal{C}_{m\_in}$ as the achievable rate region when the transmitters
ignore the non-causal state information. Compared with
$\hat{\mathcal{C}}_{m1}$ and $\mathcal{C}_{m2}$ (only utilizing DPC), we see
that the knowledge of the state information at the transmitters enlarges the
achievable rate region significantly due to DPC. Furthermore, it can be easily
seen that $\hat{\mathcal{C}}_{m3}$ and $\hat{\mathcal{C}}_{m4}$ (utilizing DPC
and active interference cancellation) are much larger than
$\hat{\mathcal{C}}_{m1}$ and $\mathcal{C}_{m2}$, respectively, which implies
that active interference cancellation improves the performance significantly.
\begin{figure}[!t]
\centering
\includegraphics[width=3.5in]{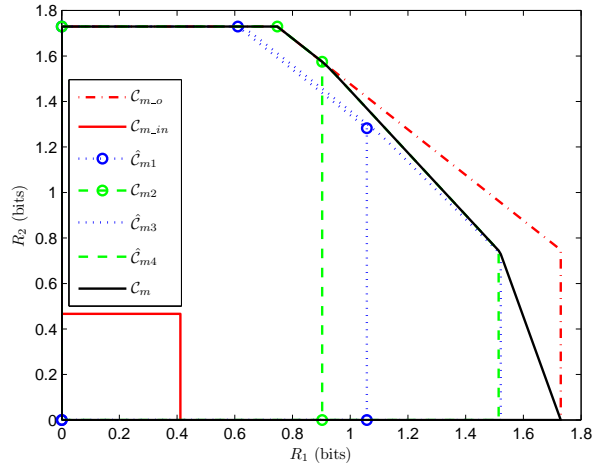}
\caption{Comparison of different achievable rate regions and the outer bound
for the mixed Gaussian IC with state information. The channel parameters are
set as: $g_{12}=0.2$, $g_{21} = 2$, $N_1=N_2=1$, $P_1=P_2=K=10\textrm{ dB}$.}
\label{fig_mixed_ic_1}
\end{figure}

For the degraded Gaussian IC with state information, we compare the achievable
rate regions with the outer bound $\mathcal{C}_{m\_o}$ and the inner bound
$\mathcal{C}_{m\_in}$ in Fig. \ref{fig_degraded_ic_1}. Note that the difference from the general mixed interference case is that the outer bound $\mathcal{C}_{m\_o}$ now includes the sum capacity~\cite{sason}. Similar to the general mixed interference case, active interference cancellation improves the performance significantly when the interference is degraded.
\begin{figure}[!t]
\centering
\includegraphics[width=3.5in]{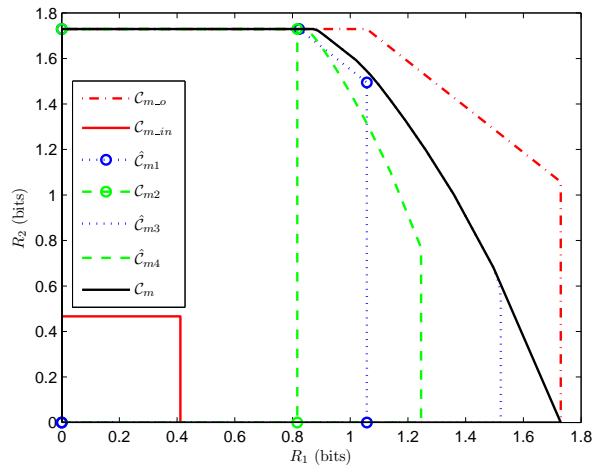}
\caption{Comparison of different achievable rate regions and the outer bound
for the degraded Gaussian IC with state information. The channel parameters
are set as: $g_{12}=0.2$, $g_{21} = 5$, $N_1=N_2=1$, $P_1=P_2=K=10\textrm{
dB}$.} \label{fig_degraded_ic_1}
\end{figure}
\subsection{Weak Gaussian IC with State Information}\label{sec_9_weak}
In Fig. \ref{fig_weak_ic_1}, we compare the achievable rate regions in
Section~\ref{sec_7_with_cancellation} with the outer bound
$\mathcal{C}_{w\_o}$, which is the same as the outer bound derived for the
traditional weak Gaussian IC~\cite{david}. Also define the inner bound
$\mathcal{C}_{w\_in}$ as the achievable rate region when the transmitters
ignore the non-causal state information. Compared with $\mathcal{C}_{w1}$ and
$\mathcal{C}_{w2}$ (only utilizing DPC), we see that the knowledge of the
state information at the transmitters improves the performance significantly
due to DPC. However, $\hat{\mathcal{C}}_{w3}$ and $\hat{\mathcal{C}}_{w4}$
(utilizing DPC and active interference cancellation) are only slightly larger
than $\mathcal{C}_{w1}$ and $\mathcal{C}_{w2}$, i.e., unlike the strong
interference case and the mixed interference case, active interference
cancellation cannot enlarge the achievable rate region significantly for the
weak interference case. Intuitively, the reason is that the source power is
too ``precious" to be used for canceling the state effect if the interference
is weak.
\begin{figure}[!t]
\centering
\includegraphics[width=3.5in]{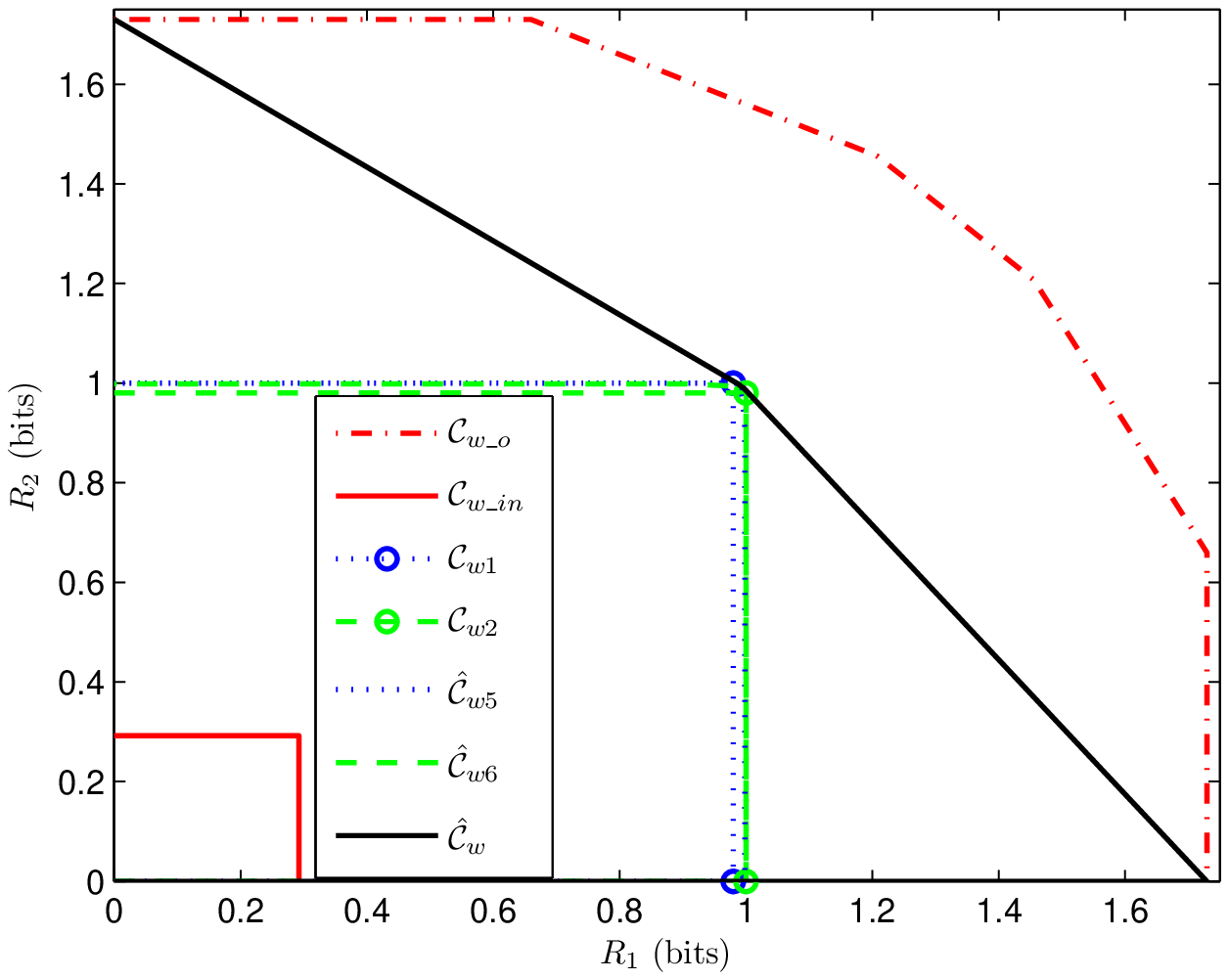}
\caption{Comparison of different achievable rate regions and the outer bound
for the weak interference Gaussian IC with state information. The channel
parameters are set as: $g_{12}=g_{21}=0.2$, $N_1=N_2=1$, $P_1=P_2=K=10\textrm{
dB}$.} \label{fig_weak_ic_1}
\end{figure}

In Fig. \ref{fig_weak_ic_2}, we compare the achievable rate regions of the
flexible power allocation schemes in Section~\ref{sec_7_with_power_allocation}
with the outer bound $\mathcal{C}_{w\_o}$ and the inner bound
$\mathcal{C}_{w\_in}$. It can be easily seen that $\hat{\mathcal{C}}_{w5}$ and
$\hat{\mathcal{C}}_{w6}$ (both utilizing DPC and flexible power allocation)
are much larger than $\mathcal{C}_{w1}$ and $\mathcal{C}_{w2}$, respectively,
i.e., flexible power allocation between the common and private messages
enlarges the achievable rate region significantly for the weak interference
case.
\begin{figure}[!t]
\centering
\includegraphics[width=3.5in]{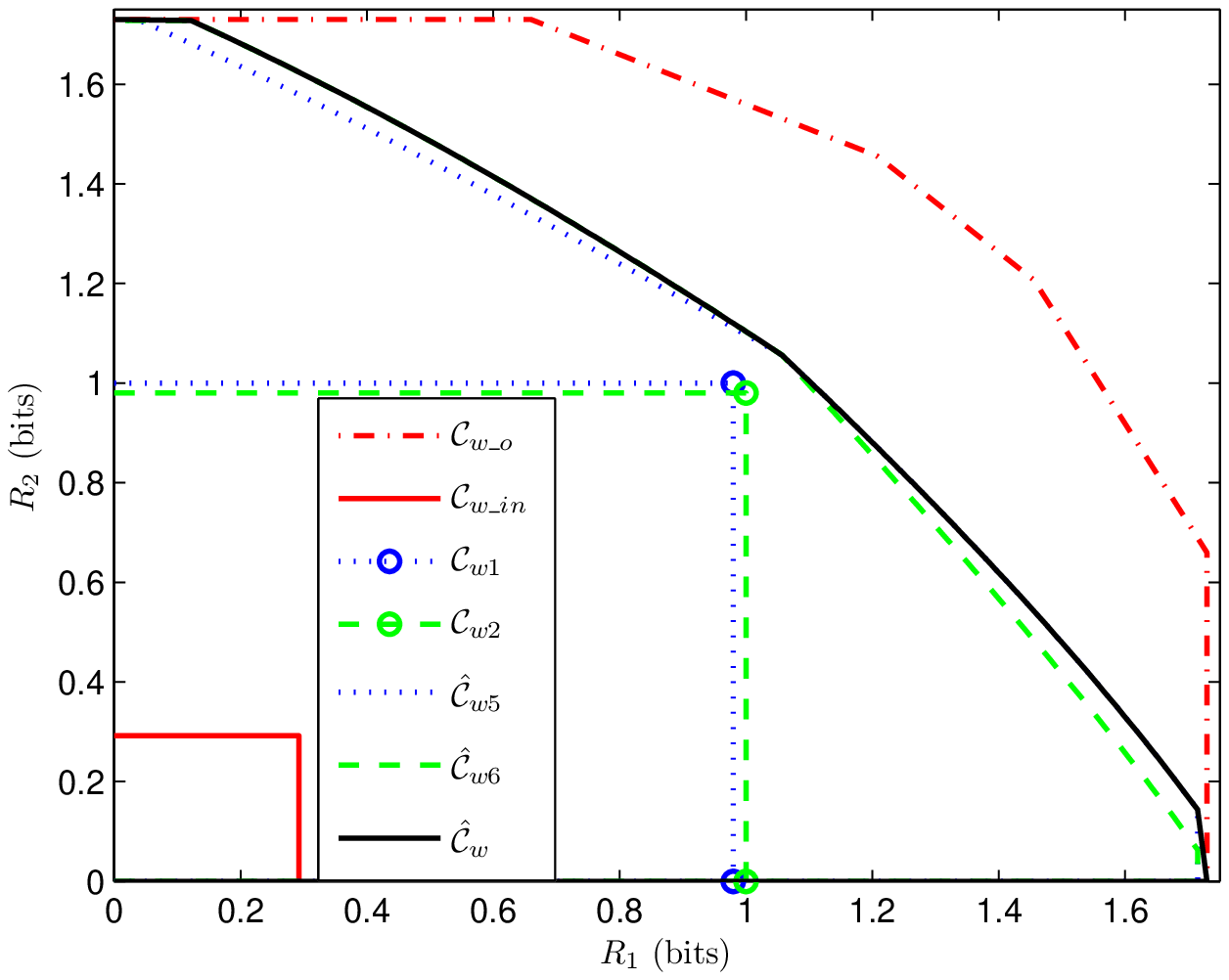}
\caption{Comparison of different achievable rate regions and the outer bound
for the weak interference Gaussian IC with state information. The channel
parameters are set as: $g_{12}=g_{21}=0.2$, $N_1=N_2=1$, $P_1=P_2=K=10\textrm{
dB}$.} \label{fig_weak_ic_2}
\end{figure}

\section{Conclusion}\label{sec_conclusion}
We considered the interference channel with state information non-causally
known at both transmitters. Two achievable rate regions were established for
the general cases based on two coding schemes with simultaneous encoding and
superposition encoding, respectively. We also studied the corresponding
Gaussian case and proposed the \emph{active interference cancellation}
mechanism, which generalizes the dirty paper coding technique, to partially
eliminate the state effect at the receivers. Several achievable schemes were
proposed and the corresponding achievable rate regions were derived for
the strong interference case, the mixed interference case, and the weak
interference case. The numerical results showed that active
interference cancellation significantly improves the performance for the strong
and mixed interference case, and flexible power splitting significantly
enlarges the achievable rate region for the weak interference case. 
\appendices
\section{Proof For Theorem \ref{theorem_1}}\label{appendix_1}
The achievable coding scheme for Theorem \ref{theorem_1} can be described as follows:

Codebook generation: Fix the probability distribution
$p(q)p(u_1|q,s)p(v_1|q,s)p(u_2|q,s)p(v_2|q,s)$. Also define the following
function for the $j$th user that maps $\mathcal{U}_j\mathcal{\times
V}_j\mathcal{\times S}$ to $\mathcal{X}_j$:
\[
x_{ji}=F_j(u_{ji},v_{ji},s_i),
\]
where $i$ is the element index of each sequence.

First generate the time-sharing sequence $q^n \sim \prod_{i=1}^{n}
p_{Q}(q_i)$. For the $j$th user, $u_{j}^{n}(m_{j0},l_{j0})$ is randomly and
conditionally independently generated according to
$\prod_{i=1}^{n}p_{U_j|Q}(u_{ji}|q_i)$, for $m_{j0} \in
\{1,2,\cdots,2^{nR_{j0}}\}$ and $l_{j0} \in \{1,2,\cdots,2^{nR_{j0}'}\}$.
Similarly, $v_{j}^{n}(m_{jj},l_{jj})$ is randomly and conditionally
independently generated according to $\prod_{i=1}^{n}p_{V_j|Q}(v_{ji}|q_i)$,
for $m_{jj} \in \{1,2,\cdots,2^{nR_{jj}}\}$ and $l_{jj} \in
\{1,2,\cdots,2^{nR_{jj}'}\}$.

Encoding: To send the message $m_j=(m_{j0},m_{jj})$, the $j$th encoder first
tries to find the pair $(l_{j0},l_{jj})$ such that the following joint
typicality holds: $(q^n,u_j^n(m_{j0},l_{j0}),s^n) \in T_\epsilon^{(n)}$ and
$(q^n,v_j^n(m_{jj},l_{jj}),s^n) \in T_\epsilon^{(n)}$. If successful,
$(q^n,u_j^n(m_{j0},l_{j0}),v_j^n(m_{jj},l_{jj}),s^n)$ is also jointly typical
with high probability, and the $j$th encoder sends $x_j$ where the $i$th
element is $x_{ji} = F_j(u_{ji}(m_{j0},l_{j0}),v_{ji}(m_{jj},l_{jj}),s_{i})$.
If not, the $j$th encoder transmits $x_j$ where the $i$th element is $x_{ji} =
F_j(u_{ji}(m_{j0},1),v_{ji}(m_{jj},1),s_{i})$.

Decoding: Decoder $1$ finds the unique message pair
$(\hat{m}_{10},\hat{m}_{11})$ such that
$(q^n,u_1^n(\hat{m}_{10},\hat{l}_{10}),u_2^n(\hat{m}_{20},\hat{l}_{20}),v_{1}^{n}(\hat{m}_{11},\\\hat{l}_{11}),y_1^n)
\in T_\epsilon^{(n)}$ for some $\hat{l}_{10} \in \{1,2,\cdots,2^{nR_{10}'}\}$,
$\hat{m}_{20} \in \{1,2,\cdots,2^{nR_{20}}\}$,
$\hat{l}_{20}\in\{1,2,\cdots,2^{nR_{20}'}\}$, and $\hat{l}_{11}\in
\{1,2,\cdots,2^{nR_{11}'}\}$. If no such unique pair exists, the decoder
declares an error. Decoder $2$ determines the unique message pair
$(\hat{m}_{20},\hat{m}_{22})$ in a similar way.

Analysis of probability of error: Here the probability of error is the same
for each message pair since the transmitted message pair is chosen with a
uniform distribution over the message set. Without loss of generality, we
assume $(1,1)$ for user $1$ and $(1,1)$ for user $2$ are sent over the
channel. First, we consider the encoding error probability at transmitter $1$.
Define the following error events:
\setlength{\arraycolsep}{2pt}{\small\begin{eqnarray*} \xi_1 &=&
\left\{\left(q^n,u_1^n\left(1,l_{10}\right),s^n\right) \notin T_\epsilon^{(n)}
\textrm{ for all } l_{10} \in
\{1,2,\cdots,2^{nR_{10}'}\}\right\},\\
\xi_2 &=& \left\{\left(q^n,v_1^n\left(1,l_{11}\right),s^n\right) \notin
T_\epsilon^{(n)} \textrm{ for all } l_{11} \in
\{1,2,\cdots,2^{nR_{11}'}\}\right\}.
\end{eqnarray*}}

The probability of the error event $\xi_1$ can be bounded as follows:
\begin{eqnarray*}
P(\xi_1) &=& \prod_{l_{10}=1}^{2^{nR_{10}'}} \left(1-P\left(\left\{\left(q^n,u_1^n\left(1,l_{10}\right),s^n\right) \in T_\epsilon^{(n)}\right\}\right)\right)\\
&\leq& \left(1-2^{-n\left(I(U_1;S|Q)+\delta_1(\epsilon)\right)}\right)^{2^{nR_{10}'}}\\
&\leq& e^{-2^{n\left(R_{10}'-I(U_1;S|Q)+\delta_1(\epsilon)\right)}},
\end{eqnarray*}
where $\delta_1(\epsilon)\to 0$ as $\epsilon \to 0$. Therefore, the
probability of $\xi_1$ goes to $0$ as $n\to\infty$ if
\begin{equation}\label{eq_prob_xi1}
R_{10}' \geq I(U_1;S|Q).
\end{equation}
Similarly, the probability of $\xi_2$ can also be upper-bounded by an
arbitrarily small number as $n\to\infty$ if
\begin{equation}\label{eq_prob_xi2}
R_{11}' \geq I(V_1;S|Q).
\end{equation}

The encoding error probability at transmitter $1$ can be calculated as:
\[
P_{\textrm{enc}1} = P\left(\xi_1 \cup \xi_2\right) \leq P(\xi_1) + P(\xi_2),
\]
which goes to $0$ as $n\to \infty$ if \eqref{eq_prob_xi1} and \eqref{eq_prob_xi2} are satisfied.

Now we consider the error analysis at decoder $1$. Denote the right
Gel'fand-Pinsker coding indices chosen by the encoders as $(L_{10},L_{11})$
and $(L_{20},L_{22})$. Define the following error events:
\setlength{\arraycolsep}{2pt}{\small
\begin{eqnarray*}
\xi_{31} &=&  \left\{\left(q^n,u_1^n\left(1,L_{10}\right),u_2^n\left(1,L_{20}\right),v_1^n\left(m_{11},l_{11}\right),y_1^n\right) \in T_\epsilon^{(n)} \textrm{ for } m_{11} \neq 1, \textrm{ and some } l_{11}\right\},\\
\xi_{32} &=&  \left\{\left(q^n,u_1^n\left(1,L_{10}\right),u_2^n\left(1,l_{20}\right),v_1^n\left(m_{11},l_{11}\right),y_1^n\right) \in T_\epsilon^{(n)} \textrm{ for } m_{11} \neq 1, \textrm{ and some } l_{11},l_{20}\neq L_{20}\right\},\\
\xi_{33} &=&  \left\{\left(q^n,u_1^n\left(1,l_{10}\right),u_2^n\left(1,L_{20}\right),v_1^n\left(m_{11},l_{11}\right),y_1^n\right) \in T_\epsilon^{(n)} \textrm{ for } m_{11} \neq 1, \textrm{ and some } l_{11},l_{10}\neq L_{10}\right\},\\
\xi_{34} &=&  \left\{\left(q^n,u_1^n\left(1,l_{10}\right),u_2^n\left(1,l_{20}\right),v_1^n\left(m_{11},l_{11}\right),y_1^n\right) \in T_\epsilon^{(n)} \textrm{ for } m_{11} \neq 1, \textrm{ and some } l_{11},l_{10}\neq L_{10},l_{20}\neq L_{20}\right\},\\
\xi_{41} &=& \left\{\left(q^n,u_1^n\left(m_{10},l_{10}\right),u_2^n\left(1,L_{20}\right),v_1^n\left(1,L_{11}\right),y_1^n\right) \in T_\epsilon^{(n)} \textrm{ for } m_{10} \neq 1, \textrm{ and some } l_{10}\right\},\\
\xi_{42} &=& \left\{\left(q^n,u_1^n\left(m_{10},l_{10}\right),u_2^n\left(1,l_{20}\right),v_1^n\left(1,L_{11}\right),y_1^n\right) \in T_\epsilon^{(n)} \textrm{ for } m_{10} \neq 1, \textrm{ and some } l_{10},l_{20}\neq L_{20}\right\},\\
\xi_{43} &=& \left\{\left(q^n,u_1^n\left(m_{10},l_{10}\right),u_2^n\left(1,L_{20}\right),v_1^n\left(1,l_{11}\right),y_1^n\right) \in T_\epsilon^{(n)} \textrm{ for } m_{10} \neq 1, \textrm{ and some } l_{10},l_{11}\neq L_{11}\right\},\\
\xi_{44} &=& \left\{\left(q^n,u_1^n\left(m_{10},l_{10}\right),u_2^n\left(1,l_{20}\right),v_1^n\left(1,l_{11}\right),y_1^n\right) \in T_\epsilon^{(n)} \textrm{ for } m_{10} \neq 1, \textrm{ and some } l_{10},l_{20}\neq L_{20},l_{11}\neq L_{11}\right\},\\
\xi_{51} &=& \left\{\left(q^n,u_1^n\left(m_{10},l_{10}\right),u_2^n\left(1,L_{20}\right),v_1^n\left(m_{11},l_{11}\right),y_1^n\right) \in T_\epsilon^{(n)} \textrm{ for } m_{10} \neq 1,\ m_{11} \neq 1, \textrm{ and some } l_{10},l_{11}\right\},\\
\xi_{52} &=& \left\{\left(q^n,u_1^n\left(m_{10},l_{10}\right),u_2^n\left(1,l_{20}\right),v_1^n\left(m_{11},l_{11}\right),y_1^n\right) \in T_\epsilon^{(n)} \textrm{ for } m_{10} \neq 1,\ m_{11} \neq 1, \textrm{ and some } l_{10},l_{11},l_{20}\neq L_{20}\right\},\\
\xi_{61} &=& \left\{\left(q^n,u_1^n\left(1,L_{10}\right),u_2^n\left(m_{20},l_{20}\right),v_1^n\left(m_{11},l_{11}\right),y_1^n\right) \in T_\epsilon^{(n)} \textrm{ for } m_{20} \neq 1,\ m_{11} \neq 1, \textrm{ and some } l_{20},l_{11}\right\},\\
\xi_{62} &=& \left\{\left(q^n,u_1^n\left(1,l_{10}\right),u_2^n\left(m_{20},l_{20}\right),v_1^n\left(m_{11},l_{11}\right),y_1^n\right) \in T_\epsilon^{(n)} \textrm{ for } m_{20} \neq 1,\ m_{11} \neq 1, \textrm{ and some } l_{20},l_{11},l_{10}\neq L_{10}\right\},\\
\xi_{71} &=& \left\{\left(q^n,u_1^n\left(m_{10},l_{10}\right),u_2^n\left(m_{20},l_{20}\right),v_1^n\left(1,L_{11}\right),y_1^n\right) \in T_\epsilon^{(n)} \textrm{ for } m_{10} \neq 1,\ m_{20} \neq 1, \textrm{ and some } l_{10},l_{20}\right\},\\
\xi_{72} &=& \left\{\left(q^n,u_1^n\left(m_{10},l_{10}\right),u_2^n\left(m_{20},l_{20}\right),v_1^n\left(1,l_{11}\right),y_1^n\right) \in T_\epsilon^{(n)} \textrm{ for } m_{10} \neq 1,\ m_{20} \neq 1, \textrm{ and some } l_{10},l_{20},l_{11}\neq L_{11}\right\},\\
\xi_8 &=& \big\{\left(q^n,u_1^n\left(m_{10},l_{10}\right),u_2^n\left(m_{20},l_{20}\right),v_1^n\left(m_{11},l_{11}\right),y_1^n\right) \in T_\epsilon^{(n)} \textrm{ for } m_{10} \neq 1,\ m_{20} \neq 1,\ m_{11}\neq 1, \\ && \textrm{ and some } l_{10},l_{20},l_{11}\big\}.
\end{eqnarray*}
} The probability of $\xi_{31}$ can be bounded as:
\begin{eqnarray*}
P(\xi_{31}) &=& \sum_{m_{11}=2}^{2^{nR_{11}}}\ \sum_{l_{11}=1}^{2^{R_{11}'}} P\left(\{\left(q^n,u_1^n\left(1,L_{10}\right),u_2^n\left(1,L_{20}\right),v_1^n\left(m_{11},l_{11}\right),y_1^n\right) \in T_\epsilon^{(n)}\}\right)\\
&\leq& 2^{n\left(R_{11}+R_{11}'\right)} \sum_{(q^n,u_1^n,u_2^n,v_1^n,y_1^n)\in T_\epsilon^{(n)}} p(q^n)p(u_1^n|q^n)p(u_2^n|q^n)p(v_1^n|q^n)p(y_1^n|u_1^n,u_2^n,q^n)\\
&\leq& 2^{n\left(R_{11}+R_{11}'\right)} 2^{-n\left(H(Q)+H(U_1|Q)+H(U_2|Q)+H(V_1|Q)+H(Y_1|U_1,U_2,Q)-H(Q,U_1,U_2,V_1,Y_1)-\delta_2(\epsilon)\right)}\\
&\leq& 2^{n\left(R_{11}+R_{11}'\right)} 2^{-n\left(I(U_1;U_2|Q)+I(U_1,U_2;V_1|Q)+I(V_1;Y_1|U_1,U_2,Q)-\delta_2(\epsilon)\right)},
\end{eqnarray*}
where $\delta_2(\epsilon)\to 0$ as $\epsilon \to 0$. Obviously, the
probability that $\xi_{31}$ happens goes to $0$ if
\begin{equation}\label{eq_prob_xi31}
R_{11}+R_{11}' \leq I(U_1;U_2|Q)+I(U_1,U_2;V_1|Q)+I(V_1;Y_1|U_1,U_2,Q).
\end{equation}
Similarly, the error probability corresponding to the other
error events goes to $0$, if \setlength{\arraycolsep}{2pt}{\small
\begin{eqnarray}
\label{eq_prob_xi32} R_{11}+R_{11}'+R_{20}' &\leq& I(U_1;U_2|Q)+I(U_1,U_2;V_1|Q)+I(V_1,U_2;Y_1|U_1,Q),\\
\label{eq_prob_xi33} R_{11}+R_{10}'+R_{11}' &\leq& I(U_1;U_2|Q)+I(U_1,U_2;V_1|Q)+I(U_1,V_1;Y_1|U_2,Q),\\
\label{eq_prob_xi34} R_{11}+R_{10}'+R_{11}'+R_{20}' &\leq& I(U_1;U_2|Q)+I(U_1,U_2;V_1|Q)+I(U_1,V_1,U_2;Y_1|Q),\\
\label{eq_prob_xi41} R_{10}+R_{10}' &\leq& I(U_1;U_2|Q)+I(U_1,U_2;V_1|Q)+I(U_1;Y_1|V_1,U_2,Q),\\
\label{eq_prob_xi42} R_{10}+R_{10}'+R_{20}' &\leq& I(U_1;U_2|Q)+I(U_1,U_2;V_1|Q)+I(U_1,U_2;Y_1|V_1,Q),\\
\label{eq_prob_xi43} R_{10}+R_{10}'+R_{11}' &\leq& I(U_1;U_2|Q)+I(U_1,U_2;V_1|Q)+I(U_1,V_1;Y_1|U_2,Q),\\
\label{eq_prob_xi44} R_{10}+R_{10}'+R_{11}'+R_{20}' &\leq& I(U_1;U_2|Q)+I(U_1,U_2;V_1|Q)+I(U_1,V_1,U_2;Y_1|Q),\\
\label{eq_prob_xi51} R_{10}+R_{11}+R_{10}'+R_{11}' &\leq& I(U_1;U_2|Q)+I(U_1,U_2;V_1|Q)+I(U_1,V_1;Y_1|U_2,Q), \\
\label{eq_prob_xi52} R_{10}+R_{11}+R_{10}'+R_{11}'+R_{20}' &\leq& I(U_1;U_2|Q)+I(U_1,U_2;V_1|Q)+I(U_1,V_1,U_2;Y_1|Q), \\
\label{eq_prob_xi61} R_{11}+R_{20}+R_{11}'+R_{20}' &\leq& I(U_1;U_2|Q)+I(U_1,U_2;V_1|Q)+I(V_1,U_2;Y_1|U_1,Q), \\
\label{eq_prob_xi62} R_{11}+R_{20}+R_{10}'+R_{11}'+R_{20}' &\leq& I(U_1;U_2|Q)+I(U_1,U_2;V_1|Q)+I(U_1,V_1,U_2;Y_1|Q), \\
\label{eq_prob_xi71} R_{10}+R_{20}+R_{10}'+R_{20}' &\leq& I(U_1;U_2|Q)+I(U_1,U_2;V_1|Q)+I(U_1,U_2;Y_1|V_1,Q), \\
\label{eq_prob_xi72} R_{10}+R_{20}+R_{10}'+R_{11}'+R_{20}' &\leq& I(U_1;U_2|Q)+I(U_1,U_2;V_1|Q)+I(U_1,V_1,U_2;Y_1|Q), \\
\label{eq_prob_xi8} R_{10}+R_{11}+R_{20}+R_{10}'+R_{11}'+R_{20}' &\leq& I(U_1;U_2|Q)+I(U_1,U_2;V_1|Q)+I(U_1,V_1,U_2;Y_1|Q).
\end{eqnarray}}
Note that there are some redundant inequalities in
\eqref{eq_prob_xi31}-\eqref{eq_prob_xi8}: \eqref{eq_prob_xi32} is implied by
\eqref{eq_prob_xi61}; \eqref{eq_prob_xi33} is implied by \eqref{eq_prob_xi51};
\eqref{eq_prob_xi42} is implied by \eqref{eq_prob_xi71}; \eqref{eq_prob_xi43}
is implied by \eqref{eq_prob_xi51}; \eqref{eq_prob_xi34},
\eqref{eq_prob_xi44}, \eqref{eq_prob_xi52}, \eqref{eq_prob_xi62}, and
\eqref{eq_prob_xi72} are implied by \eqref{eq_prob_xi8}. By combining with the
error analysis at the encoder, we can recast the rate constraints
\eqref{eq_prob_xi31}-\eqref{eq_prob_xi8} as:
\setlength{\arraycolsep}{2pt}{\small
\begin{eqnarray*}
R_{11} &\leq& I(U_1;U_2|Q)+I(U_1,U_2;V_1|Q)+I(V_1;Y_1|U_1,U_2,Q)-I(V_1;S|Q),\\
R_{10} &\leq& I(U_1;U_2|Q)+I(U_1,U_2;V_1|Q)+I(U_1;Y_1|V_1,U_2,Q)-I(U_1;S|Q),\\
R_{10}+R_{11} &\leq& I(U_1;U_2|Q)+I(U_1,U_2;V_1|Q)+I(U_1,V_1;Y_1|U_2,Q)-I(U_1;S|Q)-I(V_1;S|Q),\\
R_{11}+R_{20} &\leq& I(U_1;U_2|Q)+I(U_1,U_2;V_1|Q)+I(V_1,U_2;Y_1|U_1,Q)-I(V_1;S|Q)-I(U_2;S|Q),\\
R_{10}+R_{20} &\leq& I(U_1;U_2|Q)+I(U_1,U_2;V_1|Q)+I(U_1,U_2;Y_1|V_1,Q)-I(U_1;S|Q)-I(U_2;S|Q),\\
R_{10}+R_{11}+R_{20} &\leq& I(U_1;U_2|Q)+I(U_1,U_2;V_1|Q)+I(U_1,V_1,U_2;Y_1|Q)-I(U_1;S|Q)-I(V_1;S|Q)-I(U_2;S|Q).
\end{eqnarray*}}

The error analysis for transmitter $2$ and decoder $2$ is similar to the above
procedures and is omitted here. Correspondingly, \eqref{eq_rate_constraint_21}
to \eqref{eq_rate_constraint_26} show the rate constraints for user $2$.
In addition, the right sides of the inequalities \eqref{eq_rate_constraint_11}
to \eqref{eq_rate_constraint_26} are guaranteed to be non-negative when
choosing the probability distribution. As long as
\eqref{eq_rate_constraint_11} to \eqref{eq_rate_constraint_26} are satisfied,
the probability of error can be bounded by the sum of the error probability at
the encoders and the decoders, which goes to $0$ as $n\to\infty$.

\section{Proof For Theorem \ref{theorem_2}}\label{appendix_2}
The achievable coding scheme for Theorem \ref{theorem_2} can be described as follows:

Codebook generation: Fix the probability distribution
$p(q)p(u_1|s,q)p(v_1|u_1,s,q)p(u_2|s,q)p(v_2|u_2,s,q)$. First generate the
time-sharing sequence $q^n\sim\prod_{i=1}^{n} p_{Q}(q_i)$. For the $j$th user,
$u_j^{n}(m_{j0},l_{j0})$ is randomly and conditionally independently generated
according to $\prod_{i=1}^{n}p_{U_j|Q}(u_{ji}|q_i)$, for $m_{j0} \in
\{1,2,\cdots,2^{nR_{j0}}\}$ and $l_{j0}\in \{1,2,\cdots,2^{nR_{j0}'}\}$. For
each $u_j^n(m_{j0},l_{j0})$, $v_{j}^{n}(m_{j0},l_{j0},m_{jj},l_{jj})$ is
randomly and conditionally independently generated according to
$\prod_{i=1}^{n}p_{V_j|U_j,Q}(v_{ji}|u_{ji},q_{i})$, for $m_{jj} \in
\{1,2,\cdots,2^{nR_{jj}}\}$ and $l_{jj} \in \{1,2,\cdots,2^{nR_{jj}'}\}$.

Encoding: To send the message $m_j=(m_{j0},m_{jj})$, the $j$th encoder first
tries to find $l_{j0}$ such that $(q^n,u_j^n(m_{j0},l_{j0}),\\s^n) \in
T_\epsilon^{(n)}$ holds. Then for this specific $l_{j0}$, find $l_{jj}$ such
that $(q^n,u_j^n(m_{j0},l_{j0}),v_j^n(m_{j0},l_{j0},m_{jj},l_{jj}),s^n) \in
T_\epsilon^{(n)}$ holds. If successful, the $j$th encoder sends
$v_{j}^{n}(m_{j0},l_{j0},m_{jj},l_{jj})$. If not, the $j$th encoder transmits
$v_{j}^{n}(m_{j0},1,m_{jj},1)$.

Decoding: Decoder $1$ finds the unique message pair
$(\hat{m}_{10},\hat{m}_{11})$ such that
$(q^n,u_1^n(\hat{m}_{10},\hat{l}_{10}),u_2^n(\hat{m}_{20},\hat{l}_{20}),v_1^n(\hat{m}_{10},\\\hat{l}_{10},\hat{m}_{11},\hat{l}_{11}),y_1^n)
\in T_{\epsilon}^{(n)}$ for some $\hat{l}_{10} \in
\{1,2,\cdots,2^{nR_{10}'}\}$, $\hat{m}_{20} \in
\{1,2,\cdots,2^{nR_{20}}\}$,$\hat{l}_{20} \in \{1,2,\cdots,2^{nR_{20}'}\}$,
and $\hat{l}_{11} \in \{1,2,\cdots,2^{nR_{11}'}\}$. If no such unique pair
exists, the decoder declares an error. Decoder $2$ determines the unique
message pair $(\hat{m}_{20},\hat{m}_{22})$ similarly.

Analysis of probability of error: Similar to the proof in Theorem
\ref{theorem_1}, we assume message $(1,1)$ and $(1,1)$ are sent for both
transmitters. First we consider the encoding error probability at transmitter
$1$. Define the following error events:
\setlength{\arraycolsep}{2pt}{\small
\begin{eqnarray*} \xi_1' &=&
\left\{\left(q^n,u_1^n\left(1,l_{10}\right),s^n\right) \notin T_\epsilon^{(n)}
\textrm{ for all } l_{10} \in
\{1,2,\cdots,2^{nR_{10}'}\}\right\}, \\
\xi_2' &=&
\left\{\left(q^n,u_1^n(m_{10},l_{10}),v_1^n\left(1,l_{10},1,l_{11}\right),s^n\right)
\notin T_\epsilon^{(n)} \textrm{ for all } l_{11} \in
\{1,2,\cdots,2^{nR_{11}'}\} \textrm{ and previously found typical }
l_{10}\big|\bar{\xi}_1'\right\}.
\end{eqnarray*}}

The probability of the error event $\xi_1'$ can be bounded as:
\begin{eqnarray*}
P(\xi_1') &=& \prod_{l_{10}=1}^{2^{nR_{10}'}} \left(1-P\left(\left\{\left(q^n,u_1^n\left(1,l_{10}\right),s^n\right) \in T_\epsilon^{(n)}\right\}\right)\right)\\
&\leq& \left(1-2^{-n\left(I(U_1;S|Q)+\delta_1'(\epsilon)\right)}\right)^{2^{nR_{10}'}}\\
&\leq& e^{-2^{n\left(R_{10}'-I(U_1;S|Q)+\delta_1'(\epsilon)\right)}},
\end{eqnarray*}
where $\delta_1'(\epsilon)\to 0$ as $\epsilon \to 0$. Therefore, the
probability of $\xi_1'$ goes to $0$ as $n\to\infty$ if
\begin{equation}\label{eq_prob_hatxi1}
R_{10}' \geq I(U_1;S|Q).
\end{equation}
Similarly, for the previously found typical $l_{10}$, the probability of
$\xi_2'$ can be upper-bounded as:
\begin{eqnarray*}
P(\xi_2') &=& \prod_{l_{11}=1}^{2^{nR_{11}'}} \left(1-P\left(\left\{\left(q^n,u_1^n\left(1,l_{10}\right),v_1^n\left(1,l_{10},1,l_{11}\right),s^n\right) \in T_\epsilon^{(n)}\right\}\right)\right)\\
&\leq& \left(1-2^{n\left(H(Q,U_1,V_1,S)-H(Q,U_1,S)-H(V_1|U_1,Q)-\delta_2'(\epsilon)\right)}\right)^{2^{nR_{11}'}}\\
&\leq& \left(1-2^{-n\left(I(V_1;S|U_1,Q)+\delta_2'(\epsilon)\right)}\right)^{2^{nR_{11}'}}\\
&\leq& e^{-2^{n\left(R_{11}'-I(V_1;S|U_1,Q)+\delta_2'(\epsilon)\right)}},
\end{eqnarray*}
where $\delta_2'(\epsilon)\to 0$ as $\epsilon \to 0$. Therefore, the
probability of $\xi_2'$ goes to $0$ as $n\to\infty$ if
\begin{equation}\label{eq_prob_hatxi2}
R_{11}' \geq I(V_1;S|U_1,Q).
\end{equation}

The encoding error probability at transmitter $1$ can be calculated as:
\[
P_{\textrm{enc}1} = P(\xi_1') + P(\xi_2'),
\]
which goes to $0$ as $n\to \infty$ if \eqref{eq_prob_hatxi1} and
\eqref{eq_prob_hatxi2} are satisfied.

Now we consider the error analysis at the decoder $1$. Denote the right
Gel'fand-Pinsker coding indices chosen by the encoders as $(L_{10},L_{11})$
and $(L_{20},L_{22})$. Define the following error events:
\setlength{\arraycolsep}{2pt}{\small
\begin{eqnarray*}
\xi_{31}' &=&  \left\{\left(q^n,u_1^n\left(1,L_{10}\right),u_2^n\left(1,L_{20}\right),v_1^n\left(1,L_{10},m_{11},l_{11}\right),y_1^n\right) \in T_\epsilon^{(n)} \textrm{ for } m_{11} \neq 1, \textrm{ and some } l_{11}\right\},\\
\xi_{32}' &=&  \left\{\left(q^n,u_1^n\left(1,L_{10}\right),u_2^n\left(1,l_{20}\right),v_1^n\left(1,L_{10},m_{11},l_{11}\right),y_1^n\right) \in T_\epsilon^{(n)} \textrm{ for } m_{11} \neq 1, \textrm{ and some } l_{11},l_{20}\neq L_{20}\right\},\\
\xi_{33}' &=&  \left\{\left(q^n,u_1^n\left(1,l_{10}\right),u_2^n\left(1,L_{20}\right),v_1^n\left(1,l_{10},m_{11},l_{11}\right),y_1^n\right) \in T_\epsilon^{(n)} \textrm{ for } m_{11} \neq 1, \textrm{ and some } l_{11},l_{10}\neq L_{10}\right\},\\
\xi_{34}' &=&  \left\{\left(q^n,u_1^n\left(1,l_{10}\right),u_2^n\left(1,l_{20}\right),v_1^n\left(1,l_{10},m_{11},l_{11}\right),y_1^n\right) \in T_\epsilon^{(n)} \textrm{ for } m_{11} \neq 1, \textrm{ and some } l_{11},l_{10}\neq L_{10},l_{20}\neq L_{20}\right\},\\
\xi_{41}' &=& \left\{\left(q^n,u_1^n\left(m_{10},l_{10}\right),u_2^n\left(1,L_{20}\right),v_1^n\left(m_{10},l_{10},1,L_{11}\right),y_1^n\right) \in T_\epsilon^{(n)} \textrm{ for } m_{10} \neq 1, \textrm{ and some } l_{10}\right\},\\
\xi_{42}' &=& \left\{\left(q^n,u_1^n\left(m_{10},l_{10}\right),u_2^n\left(1,l_{20}\right),v_1^n\left(m_{10},l_{10},1,L_{11}\right),y_1^n\right) \in T_\epsilon^{(n)} \textrm{ for } m_{10} \neq 1, \textrm{ and some } l_{10},l_{20}\neq L_{20}\right\},\\
\xi_{43}' &=& \left\{\left(q^n,u_1^n\left(m_{10},l_{10}\right),u_2^n\left(1,L_{20}\right),v_1^n\left(m_{10},l_{10},1,l_{11}\right),y_1^n\right) \in T_\epsilon^{(n)} \textrm{ for } m_{10} \neq 1, \textrm{ and some } l_{10},l_{11}\neq L_{11}\right\},\\
\xi_{44}' &=& \left\{\left(q^n,u_1^n\left(m_{10},l_{10}\right),u_2^n\left(1,l_{20}\right),v_1^n\left(m_{10},l_{10},1,l_{11}\right),y_1^n\right) \in T_\epsilon^{(n)} \textrm{ for } m_{10} \neq 1, \textrm{ and some } l_{10},l_{20}\neq L_{20},l_{11}\neq L_{11}\right\},\\
\xi_{51}' &=& \left\{\left(q^n,u_1^n\left(m_{10},l_{10}\right),u_2^n\left(1,L_{20}\right),v_1^n\left(m_{10},l_{10},m_{11},l_{11}\right),y_1^n\right) \in T_\epsilon^{(n)} \textrm{ for } m_{10} \neq 1,\ m_{11} \neq 1, \textrm{ and some } l_{10},l_{11}\right\},\\
\xi_{52}' &=& \left\{\left(q^n,u_1^n\left(m_{10},l_{10}\right),u_2^n\left(1,l_{20}\right),v_1^n\left(m_{10},l_{10},m_{11},l_{11}\right),y_1^n\right) \in T_\epsilon^{(n)} \textrm{ for } m_{10} \neq 1,\ m_{11} \neq 1, \textrm{ and some } l_{10},l_{11},l_{20}\neq L_{20}\right\},\\
\xi_{61}' &=& \left\{\left(q^n,u_1^n\left(1,L_{10}\right),u_2^n\left(m_{20},l_{20}\right),v_1^n\left(1,L_{10},m_{11},l_{11}\right),y_1^n\right) \in T_\epsilon^{(n)} \textrm{ for } m_{20} \neq 1,\ m_{11} \neq 1, \textrm{ and some } l_{20},l_{11}\right\},\\
\xi_{62}' &=& \left\{\left(q^n,u_1^n\left(1,l_{10}\right),u_2^n\left(m_{20},l_{20}\right),v_1^n\left(1,l_{10},m_{11},l_{11}\right),y_1^n\right) \in T_\epsilon^{(n)} \textrm{ for } m_{20} \neq 1,\ m_{11} \neq 1, \textrm{ and some } l_{20},l_{11},l_{10}\neq L_{10}\right\},\\
\xi_{71}' &=& \left\{\left(q^n,u_1^n\left(m_{10},l_{10}\right),u_2^n\left(m_{20},l_{20}\right),v_1^n\left(m_{10},l_{10},1,L_{11}\right),y_1^n\right) \in T_\epsilon^{(n)} \textrm{ for } m_{10} \neq 1,\ m_{20} \neq 1, \textrm{ and some } l_{10},l_{20}\right\},\\
\xi_{72}' &=& \left\{\left(q^n,u_1^n\left(m_{10},l_{10}\right),u_2^n\left(m_{20},l_{20}\right),v_1^n\left(m_{10},l_{10},1,l_{11}\right),y_1^n\right) \in T_\epsilon^{(n)} \textrm{ for } m_{10} \neq 1,\ m_{20} \neq 1, \textrm{ and some } l_{10},l_{20},l_{11}\neq L_{11}\right\},\\
\xi_8' &=& \big\{\left(q^n,u_1^n\left(m_{10},l_{10}\right),u_2^n\left(m_{20},l_{20}\right),v_1^n\left(m_{10},l_{10},m_{11},l_{11}\right),y_1^n\right) \in T_\epsilon^{(n)} \textrm{ for } m_{10} \neq 1,\ m_{20} \neq 1,\ m_{11}\neq 1, \\ && \textrm{ and some } l_{10},l_{20},l_{11}\big\}.
\end{eqnarray*}
}
The probability of $\xi_{31}'$ can be bounded as follows:
\begin{eqnarray*}
P(\xi_{31}') &=& \sum_{m_{11}=2}^{2^{nR_{11}}}\ \sum_{l_{11}=1}^{2^{R_{11}'}} P\left(\{\left(q^n,u_1^n\left(1,L_{10}\right),u_2^n\left(1,L_{20}\right),v_1^n\left(1,L_{10},m_{11},l_{11}\right),y_1^n\right) \in T_\epsilon^{(n)}\}\right)\\
&\leq& 2^{n\left(R_{11}+R_{11}'\right)} \sum_{(q^n,u_1^n,u_2^n,v_1^n,y_1^n)\in T_\epsilon^{(n)}} p(q^n)p(u_1^n|q^n)p(u_2^n|q^n)p(v_1^n|u_1^n,q^n)p(y_1^n|u_1^n,u_2^n,q^n)\\
&\leq& 2^{n\left(R_{11}+R_{11}'\right)} 2^{-n\left(H(Q,U_1,V_1)+H(U_2|Q)+H(Y_1|U_1,U_2,Q)-H(Q,U_1,U_2,V_1,Y_1)-\delta_3'(\epsilon)\right)}\\
&\leq& 2^{n\left(R_{11}+R_{11}'\right)} 2^{-n\left(I(U_1,V_1;U_2|Q)+I(V_1;Y_1|U_1,U_2,Q)-\delta_3'(\epsilon)\right)},
\end{eqnarray*}
where $\delta_3'(\epsilon)\to 0$ as $\epsilon \to 0$. Obviously, the probability that $\xi_{31}'$ happens goes to $0$ if
\begin{equation}\label{eq_prob_hatxi31}
R_{11}+R_{11}' \leq I(U_1,V_1;U_2|Q)+I(V_1;Y_1|U_1,U_2,Q).
\end{equation}
Similarly, the error probability corresponding to the other
error events goes to $0$, respectively, if
\setlength{\arraycolsep}{2pt}{\small
\begin{eqnarray}
\label{eq_prob_hatxi32} R_{11}+R_{11}'+R_{20}' &\leq& I(U_1,V_1;U_2|Q)+I(V_1,U_2;Y_1|U_1,Q),\\
\label{eq_prob_hatxi33} R_{11}+R_{10}'+R_{11}' &\leq& I(U_1,V_1;U_2|Q)+I(U_1,V_1;Y_1|U_2,Q),\\
\label{eq_prob_hatxi34} R_{11}+R_{10}'+R_{11}'+R_{20}' &\leq& I(U_1,V_1;U_2|Q)+I(U_1,V_1,U_2;Y_1|Q),\\
\label{eq_prob_hatxi41} R_{10}+R_{10}' &\leq& I(U_1,V_1;U_2|Q)+I(U_1,V_1;Y_1|U_2,Q),\\
\label{eq_prob_hatxi42} R_{10}+R_{10}'+R_{20}' &\leq& I(U_1,V_1;U_2|Q)+I(U_1,V_1,U_2;Y_1|Q),\\
\label{eq_prob_hatxi43} R_{10}+R_{10}'+R_{11}' &\leq& I(U_1,V_1;U_2|Q)+I(U_1,V_1;Y_1|U_2,Q),\\
\label{eq_prob_hatxi44} R_{10}+R_{10}'+R_{11}'+R_{20}' &\leq& I(U_1,V_1;U_2|Q)+I(U_1,V_1,U_2;Y_1|Q),\\
\label{eq_prob_hatxi51} R_{10}+R_{11}+R_{10}'+R_{11}' &\leq& I(U_1,V_1;U_2|Q)+I(U_1,V_1;Y_1|U_2,Q), \\
\label{eq_prob_hatxi52} R_{10}+R_{11}+R_{10}'+R_{11}'+R_{20}' &\leq& I(U_1,V_1;U_2|Q)+I(U_1,V_1,U_2;Y_1|Q), \\
\label{eq_prob_hatxi61} R_{11}+R_{20}+R_{11}'+R_{20}' &\leq& I(U_1,V_1;U_2|Q)+I(V_1,U_2;Y_1|U_1,Q), \\
\label{eq_prob_hatxi62} R_{11}+R_{20}+R_{10}'+R_{11}'+R_{20}' &\leq& I(U_1,V_1;U_2|Q)+I(U_1,V_1,U_2;Y_1|Q), \\
\label{eq_prob_hatxi71} R_{10}+R_{20}+R_{10}'+R_{20}' &\leq& I(U_1,V_1;U_2|Q)+I(U_1,V_1,U_2;Y_1|Q), \\
\label{eq_prob_hatxi72} R_{10}+R_{20}+R_{10}'+R_{11}'+R_{20}' &\leq& I(U_1,V_1;U_2|Q)+I(U_1,V_1,U_2;Y_1|Q), \\
\label{eq_prob_hatxi8} R_{10}+R_{11}+R_{20}+R_{10}'+R_{11}'+R_{20}' &\leq& I(U_1,V_1;U_2|Q)+I(U_1,V_1,U_2;Y_1|Q).
\end{eqnarray}}
Note that there are some redundant inequalities in
\eqref{eq_prob_hatxi31}-\eqref{eq_prob_hatxi8}: \eqref{eq_prob_hatxi32} is
implied by \eqref{eq_prob_hatxi61}; \eqref{eq_prob_hatxi33} is implied by
\eqref{eq_prob_hatxi51}; \eqref{eq_prob_hatxi41} is implied by
\eqref{eq_prob_hatxi43}; \eqref{eq_prob_hatxi42} is implied by
\eqref{eq_prob_hatxi71}; \eqref{eq_prob_hatxi43} is implied by
\eqref{eq_prob_hatxi51}; \eqref{eq_prob_hatxi34}, \eqref{eq_prob_hatxi44},
\eqref{eq_prob_hatxi52}, \eqref{eq_prob_hatxi62}, \eqref{eq_prob_hatxi71}, and
\eqref{eq_prob_hatxi72} are implied by \eqref{eq_prob_hatxi8}. By combining
with the error analysis at the encoder, we can recast the rate constraints
\eqref{eq_prob_hatxi31}-\eqref{eq_prob_hatxi8} as:
\setlength{\arraycolsep}{2pt}{\small
\begin{eqnarray*}
R_{11} &\leq& I(U_1,V_1;U_2|Q)+I(V_1;Y_1|U_1,U_2,Q)-I(V_1;S|U_1,Q),\\
R_{10}+R_{11} &\leq& I(U_1,V_1;U_2|Q)+I(U_1,V_1;Y_1|U_2,Q)-I(U_1,V_1;S|Q),\\
R_{11}+R_{20} &\leq& I(U_1,V_1;U_2|Q)+I(V_1,U_2;Y_1|U_1,Q)-I(V_1;S|U_1,Q)-I(U_2;S|Q),\\
R_{10}+R_{11}+R_{20} &\leq& I(U_1,V_1;U_2|Q)+I(U_1,V_1,U_2;Y_1|Q)-I(U_1,V_1;S|Q)-I(U_2;S|Q).
\end{eqnarray*}}

The error analysis for transmitter $2$ and decoder $2$ is similar to the above
procedures and is omitted here. Correspondingly,
\eqref{eq_rate_constraint_2_21} to \eqref{eq_rate_constraint_2_24} show the
rate constraints for user $2$. Furthermore, the right-hand sides of the
inequalities \eqref{eq_rate_constraint_2_11} to
\eqref{eq_rate_constraint_2_24} are guaranteed to be non-negative when
choosing the probability distribution. As long as
\eqref{eq_rate_constraint_2_11} to \eqref{eq_rate_constraint_2_24} are
satisfied, the probability of error can be bounded by the sum of the error
probability at the encoders and the decoders, which goes to $0$ as
$n\to\infty$.

\newpage

\end{document}